\documentclass{article}

\usepackage[T1]{fontenc}
\usepackage[utf8]{inputenc}
\usepackage{lmodern}

\usepackage{amsmath}
\usepackage{amsthm}
\usepackage{amssymb}
\usepackage{amsfonts}
\usepackage{graphicx}
\usepackage{stmaryrd}
\usepackage{mathpartir}
\usepackage{mdframed} 
\usepackage{xcolor} 
\usepackage{mathpartir}
\usepackage[hidelinks]{hyperref}
\usepackage[capitalise,noabbrev,nameinlink]{cleveref} 

\usepackage{enumitem}

\usepackage{listings}
\definecolor{colKeyword}{rgb}{0.0,0.0,0.0}
\definecolor{colAssert}{rgb}{0.4,0.4,0.4}

\newcommand{\blueRightarrow}{{\color{colKeyword}$\Rightarrow$}}
\newcommand{\blueBar}{{\color{colKeyword}\texttt{|}}}
\newcommand{\blueMult}{{\color{colKeyword}$\times$}}

\lstdefinelanguage{clerical}{
  keywords={true,false,skip,lim,var,in,if,then,else,begin,end,case,while,do,let},
  literate={:=}{{\color{colKeyword}$\;{{:}{=}}\;$}}2 {=>}{\blueRightarrow}2 {|}{\blueBar}1{*}{\blueMult}1,
  keywordstyle=\color{colKeyword}\ttfamily,
  identifierstyle=\itshape,
  escapeinside={[}{]}
}

\lstdefinestyle{clerical_noline}{
  basicstyle=\color{colKeyword}\small,
  mathescape=true,
  xleftmargin=4ex,
  language=clerical,
}

\lstdefinestyle{clerical_line}{
  basicstyle=\color{colKeyword}\small,
  mathescape=true,
  xleftmargin=4ex,
  language=clerical,
  numbers=left,                    
  numbersep=5pt
}

\definecolor{dkgreen}{rgb}{0,0.6,0}
\definecolor{ltblue}{rgb}{0,0.4,0.4}
\definecolor{dkviolet}{rgb}{0.3,0,0.5}

\lstdefinelanguage{Coq}{ 
    mathescape=true,
    texcl=false, 
    morekeywords=[1]{Section, Module, End, Require, Import, Export,
        Variable, Variables, Parameter, Parameters, Axiom, Hypothesis,
        Hypotheses, Notation, Local, Tactic, Reserved, Scope, Open, Close,
        Bind, Delimit, Definition, Let, Ltac, Fixpoint, CoFixpoint, Add,
        Morphism, Relation, Implicit, Arguments, Unset, Contextual,
        Strict, Prenex, Implicits, Inductive, CoInductive, Record,
        Structure, Canonical, Coercion, Context, Class, Global, Instance,
        Program, Infix, Theorem, Lemma, Corollary, Proposition, Fact,
        Remark, Example, Proof, Goal, Save, Qed, Defined, Hint, Resolve,
        Rewrite, View, Search, Show, Print, Printing, All, Eval, Check,
        Projections, inside, outside, Def},
    morekeywords=[2]{forall, exists, exists2, fun, fix, cofix, struct,
        match, with, end, as, in, return, let, if, is, then, else, for, of,
        nosimpl, when},
    morekeywords=[3]{Type, Prop, Set, true, false, option},
    morekeywords=[4]{pose, set, move, case, elim, apply, clear, hnf,
        intro, intros, generalize, rename, pattern, after, destruct,
        induction, using, refine, inversion, injection, rewrite, congr,
        unlock, compute, ring, field, fourier, replace, fold, unfold,
        change, cutrewrite, simpl, have, suff, wlog, suffices, without,
        loss, nat_norm, assert, cut, trivial, revert, bool_congr, nat_congr,
        symmetry, transitivity, auto, split, left, right, autorewrite},
    morekeywords=[5]{by, done, exact, reflexivity, tauto, romega, omega,
        assumption, solve, contradiction, discriminate},
    morekeywords=[6]{do, last, first, try, idtac, repeat},
    morecomment=[s]{(*}{*)},
    showstringspaces=false,
    morestring=[b]",
    morestring=[d],
    tabsize=3,
    extendedchars=false,
    sensitive=true,
    breaklines=false,
    basicstyle=\small,
    captionpos=b,
    columns=[l]flexible,
    identifierstyle={\ttfamily\color{black}},
    keywordstyle=[1]{\ttfamily\color{dkviolet}},
    keywordstyle=[2]{\ttfamily\color{dkgreen}},
    keywordstyle=[3]{\ttfamily\color{ltblue}},
    keywordstyle=[4]{\ttfamily\color{dkblue}},
    keywordstyle=[5]{\ttfamily\color{dkred}},
    stringstyle=\ttfamily,
    commentstyle={\ttfamily\color{dkgreen}},
    literate=
    {\\forall}{{\color{dkgreen}{$\forall\;$}}}1
    {\\exists}{{$\exists\;$}}1
    {<-}{{$\leftarrow\;$}}1
    {||-}{{$\Vdash\;$}}1
    {|-}{{$\vdash\;$}}1
    {|=}{{$\vDash\;$}}1
    {=>}{{$\Rightarrow\;$}}1
    {==}{{\code{==}\;}}1
    {==>}{{\code{==>}\;}}1
    {->}{{$\rightarrow\;$}}1
    {<->}{{$\leftrightarrow\;$}}1
    {<==}{{$\leq\;$}}1
    {\#}{{$^\star$}}1 
    {\\o}{{$\circ\;$}}1 
    {\@}{{$\cdot$}}1 
    {\/\\}{{$\wedge\;$}}1
    {\\\/}{{$\vee\;$}}1
    {++}{{\code{++}}}1
    {~}{{\ }}1
    {\@\@}{{$@$}}1
    {\\mapsto}{{$\mapsto\;$}}1
    {\\hline}{{\rule{\linewidth}{0.5pt}}}1
}[keywords,comments,strings]

\newcommand{\coqcode}[1]{\lstinline[language=Coq]!#1!}

\newcommand{\coqpath}[1]{\texttt{#1}}

\newcommand{\defemph}[1]{\textbf{\emph{#1}}} 

\theoremstyle{plain}

\newtheorem{theorem}{Theorem}[section]

\newtheorem{proposition}[theorem]{Proposition}

\theoremstyle{definition}

\newcommand{\defeq}{\mathrel{{:}{=}}} 
\newcommand{\eqdef}{\mathrel{{=}{:}}} 
\newcommand{\defiff}{\mathrel{{:}{\Leftrightarrow}}} 

\newcommand{\dom}[1]{\mathsf{dom}(#1)} 

\newcommand{\abs}[1]{\lvert #1 \rvert}

\newcommand{\all}[1]{\forall #1 \,.\,}
\newcommand{\some}[1]{\exists #1 \,.\,}

\newcommand{\lthen}{\Rightarrow}

\newcommand{\liff}{\Leftrightarrow}

\newcommand{\IN}{\mathbb{N}}
\newcommand{\IQ}{\mathbb{Q}}
\newcommand{\IR}{\mathbb{R}}
\newcommand{\IZ}{\mathbb{Z}}

\newcommand{\numeral}[1]{\overline{#1}} 

\newcommand{\bnfis}{\mathrel{\;{:}{:}{=}\ }}
\newcommand{\bnfor}{\mathrel{\;\big|\ \ }}

\newcommand{\dR}{\mathsf{R}}
\newcommand{\dZ}{\mathsf{Z}}
\newcommand{\dB}{\mathsf{B}}
\newcommand{\dU}{\mathsf{U}}

\newcommand{\ccoerce}[1]{\iota(#1)} 
\newcommand{\ccase}{\mathtt{case}\;}
\newcommand{\cif}{\mathtt{if}\;}
\newcommand{\cthen}{\;\mathtt{then}\;}
\newcommand{\celse}{\;\mathtt{else}\;}
\newcommand{\cend}{\;\mathtt{end}}
\newcommand{\cwhile}{\mathtt{while}\;}
\newcommand{\climx}{\mathtt{lim}} 
\newcommand{\clim}[2]{\climx \; #1 \,.\, #2}
\newcommand{\cdo}{\;\mathtt{do}\;}
\newcommand{\cin}{\;\mathtt{in}\;}
\newcommand{\cskip}{\mathtt{skip}}

\newcommand{\cletx}{\mathrel{{:}{=}}} 
\newcommand{\clet}[2]{#1 \cletx #2} 
\newcommand{\cvar}[2]{\mathtt{var}\; #1 \cletx #2 \cin} 
\newcommand{\cinv}[1]{{#1}^{-1}} 
\newcommand{\To}{\Rightarrow}
\newcommand{\cfunction}[4]{\mathtt{let}\; #1(#2) : #3 \cletx #4} 

\newcommand{\ctrue}{\mathsf{true}}
\newcommand{\cfalse}{\mathsf{false}}

\newcommand{\rlt}{\mathbin{\boldsymbol{<}}}
\newcommand{\rplus}{\mathbin{\boldsymbol{+}}}
\newcommand{\rminus}{\mathop{\boldsymbol{-}}}
\newcommand{\rmult}{\mathbin{\boldsymbol{\times}}}

\newcommand{\ilt}{\mathbin{<}}
\newcommand{\ieq}{\mathbin{=}}
\newcommand{\iplus}{\mathbin{+}}
\newcommand{\iminus}{\mathbin{-}}
\newcommand{\imult}{\mathbin{\times}}

\newcommand{\op}{\mathbin{\ast}} 
\newcommand{\iop}{\mathbin{\varoast}} 
\newcommand{\rop}{\mathbin{\boxast}} 

\newcommand{\cdisj}{\mathrel{\bar{\lor}}}

\newcommand{\emptyctx}{{\cdot}} 

\newcommand{\of}{{:}} 
\newcommand{\rwtypes}{\Vdash} 
\newcommand{\rotypes}{\vdash} 

\newcommand{\isEnv}{\;\mathsf{env}} 

\newcommand{\rulename}[1]{\textnormal{\textsc{#1}}}

\newcommand{\rref}[1]{\hyperlink{rule:#1}{\rulename{#1}}}

\definecolor{rulenameColor}{rgb}{0.5,0.5,0.5}

\newcommand{\inferenceRule}[3]{\inferrule*[lab={\hypertarget{rule:#1}{\rulename{\footnotesize\color{rulenameColor}#1}}}]{#2}{#3}}

\newcommand{\semtt}{\mathsf{tt}} 
\newcommand{\semff}{\mathsf{ff}} 
\newcommand{\semuu}{\star} 

\newcommand{\inclZ}[1]{\iota_{\IZ}(#1)} 

\newcommand{\sem}[1]{\llbracket #1 \rrbracket} 

\newcommand{\liftnoerr}[1]{{#1}_\bot} 

\newcommand{\Pstar}{\mathcal{P}_{\!\star}}

\newcommand{\dsup}[1]{{\textstyle\bigsqcup_{#1}}}

\newcommand{\PP}[1]{\Pstar(#1)} 
\newcommand{\PPleq}{\sqsubseteq} 
\newcommand{\pure}[1]{\{#1\}} 
\newcommand{\lift}[1]{#1^\dagger} 
\newcommand{\PPlet}[2]{\mathsf{let}\; #1 \shortleftarrow #2 \;\mathsf{in}\;} 
\newcommand{\PPletx}[1]{\mathsf{let}\; #1}
\newcommand{\PPinx}[1]{\shortleftarrow #1 \;\mathsf{in}\;} 
\newcommand{\PPtuple}[1]{\langle #1 \rangle_\star} 

\newcommand{\PPerr}{\emptyset} 
\newcommand{\PPbot}{\{\bot\}} 

\newcommand{\prt}{\mathsf{p}}
\newcommand{\tot}{\mathsf{t}}
\newcommand{\such}{\mid}
\newcommand{\fv}[1]{\mathsf{fv}(#1)} 

\newcommand{\rotrip}[5][\star]{#2 \rotypes \{ #3\} \, #4 \, \{ #5 \}^{#1}}
\newcommand{\rwtrip}[5][\star]{#2 \rwtypes \{ #3\} \, #4 \, \{ #5 \}^{#1}}

\newcommand{\claim}[1]{{\color{colAssert}$\{\,#1\,\}$}}

\newcommand{\claimp}[1]{{\color{colAssert}$\{\,#1\,\}^{\prt}$}}
\newcommand{\claimt}[1]{{\color{colAssert}$\{\,#1\,\}^{\tot}$}}

\makeatletter
\newcommand{\mathleft}{\@fleqntrue\@mathmargin0pt}
\newcommand{\mathcenter}{\@fleqnfalse}
\makeatother

\counterwithin{figure}{section}

\title{An Imperative Language for Verified Exact Real-Number Computation}
\author{
Andrej Bauer${}^{1,2}$\thanks{This material is based upon work supported by the Air Force Office of Scientific Research under award number FA9550-21-1-0024.},
Sewon Park$^3$\thanks{A preliminary version of this work is included in the authors' doctoral dissertation~\cite{sewonphd}. 
This work was supported by the National Research Foundation of Korea (NRF) funded by the Korea government(MSIT) (No. 2016K1A3A7A03950702) and by JSPS KAKENHI (Grant-in-Aid for JSPS Fellows) JP22F22071.},
Alex Simpson${}^{1,2}$\thanks{This project has received funding from the  European Union's  Horizon  2020  research  
and innovation programme under the Marie Sk{\l}odowska-Curie grant  agreement No 731143.\newline
\includegraphics[scale=0.07]{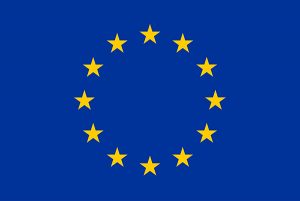}}\\
\small $^{1}$Faculty of Mathematics and Physics, University of Ljubljana -- Slovenia\\
\small $^{2}$Institute for Mathematics, Physics and Mechanics -- Slovenia\\
\small $^{3}$Graduate School of Informatics, Kyoto University -- Japan
} 

\date{}
\begin{document}

\maketitle

\begin{abstract}

We introduce \emph{Clerical}, a programming language for exact real-number computation that combines first-order imperative-style programming with
a limit operator for computation of real numbers as limits of Cauchy sequences.
We address the semidecidability of the linear ordering of the reals by incorporating nondeterministic guarded choice, through which decisions based on partial comparison operations on reals can be patched together to give total programs.
The interplay between mutable state, nondeterminism, and computation of limits is controlled by the requirement that expressions computing limits and guards modify only local state.
We devise a domain-theoretic denotational semantics that uses a variant of Plotkin powerdomain construction tailored to our specific version of nondeterminism.
We formulate a Hoare-style specification logic, show that it is sound for the denotational semantics,
and illustrate the setup by implementing and proving correct a program for computation of~$\pi$ as the least positive zero of~$\sin$. The modular character of Clerical allows us to compose the program from smaller parts, each of which is shown to be correct on its own.
We provide a proof-of-concept OCaml implementation of Clerical, and formally verify parts of the development, notably the soundness of specification logic, in the Coq proof assistant.
\end{abstract}

\bigskip


\clearpage

\tableofcontents

\section{Introduction}
In exact real-number computation, infinite representations are used to compute with real numbers precisely, without rounding errors.
By representing reals as, e.g., infinite sequences of rational approximations,
real-valued functions can be computed exactly, using stream algorithms or
type-2 Turing machines~\cite{w00}.
In many approaches to exact real-number computation \cite{tucker1999computation,tucker1999computation,tucker2015generalizing,escardo1996pcf,edalat2000integration,escardoSimpson2014,brausse2016semantics} the concrete representation of real numbers is veiled by an \emph{abstract datatype} or \emph{interface} that exposes only a suite of primitive operations on reals.
This way programmers can think of the real numbers in ordinary mathematical terms, as
a structure closely related to the usual one~\cite{hertling99:_real_number_struc_effec_categ,escardoSimpson2001}. Moreover, programs can be written and reasoned about intuitively, relying on familiarity with the traditional mathematics of real numbers, and without having to take rounding errors or representations into account.
This approach has been substantiated in practice by several implementations~\cite{muller2000irram,lambov2007,Ariadne,aern}. 

Imperative programming is a natural and ubiquitous programming paradigm, supported by a well-established precondition/postcondition-based program verification methodology~\cite{apt19:_fifty_years_hoare_logic}.
One would naturally like to incorporate exact real-number computation into the imperative programming style and its verification methodology. This desire has been previously addressed in \cite{brausse2016semantics}, in which the authors introduce a simple imperative language with an abstract data type of reals, provide formal rules for proving correctness assertions, and present illustrative examples, such as a verified root-finding program.

The goal of the present paper is to extend the work of \cite{brausse2016semantics} to a richer programming language, extended with a limit-finding primitive that calculates limits of Cauchy sequences, a central feature of exact real-number computation that makes it more expressive than purely algebraic or symbolic computation \cite{Bra03f,neumann2018topological}. The key operation takes a Cauchy sequence $f : \IN \to \IR$ (with a fixed rate of convergence) and returns its limit $\lim_{n \to \infty} f(n)$ as a value of real number type.
Such functionality is implemented in practice~\cite{muller2000irram}. However, it is not included in the imperative language of \cite{brausse2016semantics}, in which limits can only be implemented \emph{indirectly} as top-level programs that calculate approximations forming Cauchy sequences. Because the limit values are not themselves directly accessible, they cannot be returned or used in intermediate calculations, nor can calculations using limits be composed.
A similar restriction can be found in the work of Tucker and Zucker \cite{tucker1999computation,tucker2015generalizing}.

In this paper, we present an imperative language called 
\emph{Clerical} (Command-Like Expressions for Real Infinite-precision Calculations)  that
fully supports the construction of real numbers as limits, and we provide a program verification framework for it.

The paper is organized as follows.
In \Cref{sec:overview} we give an overview of the design challenges and the solutions employed in Clerical.
A formal account of the syntax and the type system is given in \Cref{sec:syntax}.
In \cref{sec:denotation} we introduce a modified Plotkin powerdomain and use it to define a denotational semantics of Clerical, including the relatively involved denotations of limits, guarded nondeterministic choice, and loops.
In \cref{sec:boolean-ops}, we show that Clerical is expressive enough to support parallel evaluation and nondeterministic operations, including McCarthy's ambiguous choice.
In \cref{sec:specification-logic} we define a specification logic and prove it to be sound for our denotational semantics.
The capabilities of our setup are illustrated in \cref{sec:example}, where we implement a Clerical program computing~$\pi$ as the least positive root of the $\sin$ function and show its correctness.
In \cref{sec:implementation} we briefly address the operational aspects of Clerical and how one might implement a practical programming language based on it. We provide our own proof-of-concept implementation, with several minor extensions to the language.
Finally, in \cref{sec:formalization} we comment on our formalization of soundness of specification logic and several other parts of Clerical in the proof assistant Coq~\cite{coq}.

\section{Overview of Clerical}
\label{sec:overview}

Clerical is an imperative-style programming language with mutable variables, conditional statements, loops, equipped with an abstract datatype of real numbers that supports basic arithmetic, and crucially, computation of real numbers as limits of Cauchy sequences. As the strict linear order~$<$ on the reals is only semidecidable, conditional statements need to cope with possibly nonterminating comparison tests, leading to a necessarily nondeterministic language.  
The language is designed to exploit the programming potential of interactions between mutable state, computation of limits, nondeterminism, and nontermination. Its type system and specification logic can be used to ensure that such interactions are only utilised in error-free ways.

The limit-finding operator $\lim : (\IN \to \IR) \to \IR$ could be defined naturally as a higher-order function taking a Cauchy sequence $f: \IN \to \IR$ to its limits; see \cite{park23} for an example. 
Instead, we adopt a more general approach that does not require  function types (but which is equivalent in their presence). 
We introduce a custom expression $\clim{n}{e}$ that encapsulates a real-valued expression~$e$ (usually containing $n$) computing approximations of the limit.
When the sequence of approximations converges rapidly, i.e., the $n$-th one is within $2^{-n}$ of the limit, the expression $\clim{n}{e}$ is guaranteed to evaluate to the real number that is the limit of the sequence.
For the language to be rich enough to express interesting limits, it is important that the calculation of~$e$ can be an arbitrarily complex computation that may use the full set of programming features, including sequencing, loops and mutable state. In this sense, expressions in  Clerical are \emph{command-like}, 
going far beyond the simple algebraic expressions included in the imperative languages typically used as the basis of Hoare logic~\cite{apt19:_fifty_years_hoare_logic}. Nevertheless,  
as in such standard imperative languages, we do require the 
expression~$e$ 
in $\clim{n}{e}$ to be \emph{pure}, meaning that it does not alter the global state. To ensure this, $e$ is  allowed 
to perform assignment only to its own \emph{local variables}, with all other variables accessed in a read-only fashion.
This purity condition is enforced via a typing discipline, as explained in~\cref{sec:syntax}.

Non-termination needs careful handling in exact real-number computation. 
Even a simple order comparison of two real numbers can give rise to non-termination when the numbers coincide \cite[Theorem 4.1.16]{w00}. 
That is, the test $x \rlt y$ does not terminate when $x$ and $y$ are equal.
To safely deal with real number comparisons that may diverge, nondeterminism becomes essential \cite{LUCKHARDT1977321}. 
In Clerical, nondeterminism is provided by a Dijkstra-style guarded case statement:
\[\ccase e_1  \Rightarrow c_1 \mid e_2 \Rightarrow c_2 \mid \cdots \mid e_n \Rightarrow c_n \cend\;.\]
This construct proceeds with parallel or interleaved evaluation of the guards $e_i$,
and selects one of the guarded expressions~$c_i$ whose guard~$e_i$ evaluated to~$\ctrue$.
If several guards are true, any one of the corresponding expressions may be selected nondeterministically,
which allows non-termination to be bypassed: even if one of the guards fails to terminate,
as long as there is some guarded expression that can be selected, the case expression safely selects such a branch.
For example, soft comparison \cite{BRATTKA1998490}, a nondeterministic approximation to 
order comparison, is expressed in Clerical as follows:
\[
\ccase x \rlt y \rplus 2^{-n}  \Rightarrow \ctrue \mid y \rlt x \rplus 2^{-n} \Rightarrow \cfalse \cend\;.
\]
In the case $x = y + 2^{-n}$, the first guard fails to terminate, as does the second in the case $x = y - 2^{-n}$.
The overall expression evaluates deterministically to $\ctrue$ in the case that $x \leq  y - 2^{-n}$, to 
$\cfalse$ in the case that $x \geq  y + 2^{-n}$, and nondeterministically to either truth value in the case
$y - 2^{-n} < x < y + 2^{-n}$.

The limit operator adds a further complication to Clerical.
It is not clear how to devise a coherent evaluation strategy for $\clim{n}{e}$ when the expression~$e$ fails to generate a rapidly converging Cauchy sequence.
For this reason, we consider any program that applies the limit operation to an expression $e$ that does not define a rapidly converging Cauchy sequence as being erroneous.
In the denotational semantics of Clerical, such programs will assume a special error value, which we must be able to combine semantically with non-termination and nondeterminism. For this purpose we 
use  
a modified version of the Plotkin powerdomain~\cite{plotkin1976powerdomain} in \cref{sec:denotation}.
We remark that the error value cannot be identified with non-termination, because, unlike non-termination, error values cannot be bypassed by guarded case statements.

A general (possibly impure) expression in  Clerical has two main behaviours: it evaluates to a value and also, in the case that 
it is impure, alters the state. 
Accordingly, we use the following form of Hoare-style triples, whose postconditions can refer to return values as in \cite{hot,HONDA201475,iris},
for partial and total correctness specifications for expressions:
\[
\{ \phi\} \, e\, \{ y \of \tau \such \psi \}^\tot
    \qquad\text{and}\qquad
\{ \phi\} \, e\, \{ y \of \tau \such \psi \}^\prt\;.
\]
The first triple expresses total correctness of the expression $e$ and says that,
in any state satisfying the precondition $\phi$, every nondeterministic branch in the execution of $e$ terminates without introducing any error, and each branch results in a state and a value $y$ satisfying the postcondition $\psi$.
The second triple for partial correctness still requires terminating nondeterministic branches to satisfy $\psi$ and all branches
to be  error free, but
permits the existence of nondeterministic branches that do not terminate.
Although total correctness is the form of correctness that is usually desired in program verification, it turns out  to be necessary to also consider partial correctness in order to provide the correct proof rules for total correctness.
When we prove the total correctness of a guarded case expression, it is appropriate to assume that the guards are only partially correct,
since the case expression can still  terminate when some guards do not.
Conversely, total correctness is required to formulate the correct proof rules for partial correctness.
To prove a limit operation $\clim{n}{e}$ partially correct, it is still necessary to ensure that $e$ is totally correct, as any non-terminating behaviour in $e$ will prevent it from defining a rapid Cauchy sequence, and hence $\clim{n}{e}$ to be erroneous, as discussed above. 
Accordingly, we provide proof rules for partial and total correctness that are intertwined and prove that they are sound with regard to the denotational semantics.


\section{Syntax and type system}
\label{sec:syntax}

As just discussed, Clerical is an imperative language for exact real-number computation based on \emph{command-like expressions}, that is, value-returning expressions built up using the usual constructs  for forming imperative commands: while loops, case statements, sequential composition, variable assignment, etc. Variables and values are typed. The types of Clerical are:
a  \text{unit} type, $\dU$, with only one value; booleans, $\dB$; integers, $\dZ$; and the abstract type of real numbers, $\dR$, whose 
very presence is the \emph{raison d'\^{e}tre} of the language. 

The well-typedness of expressions is ensured as usual by a type system that provides rules for 
establishing that an expression $e$ has type $\tau$ in a context that assigns types to  all the variables that appear in $e$.
There is, however, one particular subtlety that we need to deal with. As adumbrated in Section~\ref{sec:overview}, the programming language distinguishes between `pure' and general (`impure') expressions. By a \emph{pure} expression $e$, we understand one that can read from but not write to the variables in the context of the expression.  That is, $e$ must 
not contain within it any assignment operation $\clet{x}{e'}$ to a variable $x$ in the context of $e$. 
Since, in imperative programming, variable assigment is an essential component of almost any non-trivial computation, in order to permit a sufficiently expressive language of pure expressions in Clerical, we allow pure expressions to declare and assign values to 
\emph{local variables} that do not appear in the context of $e$. This is formalized via  a type system in which the general form of judgement for assigning an expression $c$ the type $\tau$ is 
\[
  \Gamma; \Delta \rwtypes c : \tau  \enspace , 
\]
where $\Gamma = (x_1 \of \tau_1, \ldots, x_m \of \tau_m)$ is a list assigning types to variables that are read-only in~$c$,
and $\Delta = (y_1 \of \sigma_1, \ldots, y_n \of \sigma_n)$ assigns types to read-write variables. 
Pure expressions arise by requiring the context $\Delta$ of read-write variables to be empty. For convenience in defining the semantics of 
Clerical in~\cref{sec:denotation},
we include a separate judgement form for pure expressions:
\[
  \Gamma \rotypes e : \tau  \enspace , 
\]
in which $\Gamma$, as above,   is a context of read-only variables. In general, we refer to sequences $\Gamma$ as 
\emph{read-only contexts}, and pairs
\begin{equation*}
 \Gamma ; \Delta =
 (x_1 \of \tau_1, \ldots, x_m \of \tau_m ;\; y_1 \of \sigma_1, \ldots, y_n \of \sigma_n)
\end{equation*}
as \emph{read-write contexts}. The two judgement forms are connected by conversion  rules in each direction:
\[
   \inferenceRule{}{
      \Gamma ; \emptyctx \rwtypes c : \tau
    }{
      \Gamma \rotypes c : \tau
    }
\qquad
    \inferenceRule{}{
      \Gamma, \Delta \rotypes e : \tau
    }{
      \Gamma ; \Delta \rwtypes e : \tau
    }
\]
The first explicitly recognises the purity of a general expression with an empty sequence of write variables. The second allows us to use a pure expression
as a general expression, by choosing any partition of the \emph{read-only} context $\Gamma,\Delta$ to identify the variables
$\Delta$ that we allow the general expression to write to. 
The distinction between~$c$ as the meta-notation for general expressions and~$e$ as the meta-notation for pure expressions reflects the fact that, in Clerical, we use general expressions~$c$ in syntactic positions in which `commands' usually appear in traditional imperative languages, and we use pure expressions~$e$ in positions that are traditionally restricted to  `expressions' in imperative language nomenclature. Indeed, our general value-returning expressions generalise the traditional value-free `commands' since we can consider the latter as being given by
general expressions of unit type in read-write contexts in which all variables are writable:
\[
  \emptyctx  \, ; \Delta \rwtypes c : \dU \enspace . 
\]
Further mediation between the collections $\Gamma$ of read-only variables and $\Delta$ of read-write variables is provided by the typing rule for local variables
\[
 \inferenceRule{}{
      \Gamma, \Delta \rotypes e : \sigma \\
      \Gamma ; \Delta, x \of \sigma \rwtypes c : \tau
     }{
      \Gamma ; \Delta \rwtypes (\cvar{x}{e} c) : \tau
    }
\]
The expression $\cvar{x}{e} c$ declares a fresh local variable $x$ and initializes it by computing the value of the subexpression~$e$, which is required to be pure. (This is an example of a syntactic position that would be occupied by an `expression' in a traditional imperative language, and which accordingly  requires a pure expression in Clerical.) The main expression then continues computation as the general expression~$c$, which is at liberty to overwrite the variable $x$ without affecting the status (read-only or read-write respectively) of the variables in~$\Gamma$ and~$\Delta$.

\begin{figure}
  \centering
  \begin{mdframed}
  \small
  \begin{align*}
  \text{Expression}\ e, c
  \bnfis& x                               &&\text{variable}\\
  \bnfor& \ctrue \bnfor \cfalse \bnfor \numeral{k} \bnfor \cskip
                                          &&\text{constants}\\
  \bnfor& \ccoerce{e}                     &&\text{coercions from $\dZ$}\\ 
  & &&\quad \text{to $\dR$}\\
  \bnfor& 2^e                             &&\text{exponentiation by $2$}\\
  \bnfor& e_1 \iop e_2                    &&\text{integer arithmetic} \\ 
  & &&\quad\text{${\iop} \in \{{\iplus}, {\iminus}, {\imult}\}$}\\
  \bnfor& e_1 \rop e_2 \bnfor \cinv{e}    &&\text{real arithmetic} \\
  & &&\quad\text{${\rop} \in \{{\rplus}, {\rminus}, {\rmult}\}$} \\
  \bnfor& e_1 \ilt e_2 \bnfor e_1 = e_2   &&\text{integer comparison}\\
  \bnfor& e_1 \rlt e_2                    &&\text{real comparison}\\
  \bnfor& \clim{x}{e}                     &&\text{limit ($x$ bound in $e$)} \\
  \bnfor& c_1 ; c_2                       &&\text{sequencing}\\
  \bnfor& \cvar{x}{e} c                   &&\text{local variable}\\
   & && \quad\text{($x$ bound in $c$)}\\
  \bnfor& \clet{x}{e}                     &&\text{assignment}\\
  \bnfor& \cif e \cthen c_1 \celse c_2 \cend
                                          &&\text{conditional}\\
  \bnfor& \ccase e_1 \To c_1 \mid \cdots \mid e_n \To c_n \cend
                                          &&\text{guarded cases}\\
  \bnfor& \cwhile e \cdo c \cend          &&\text{loop} \\[1ex]
  \text{Type}\ \tau, \sigma
  \bnfis& \dU                             &&\text{unit}\\
  \bnfor& \dB                             &&\text{boolean}\\
  \bnfor& \dZ                             &&\text{integer}\\
  \bnfor& \dR                             &&\text{real} \\[1ex]
  \text{Typing context}\ \Gamma, \Delta
  \bnfis& \ x_1 \of \tau_1, \ldots, x_m \of \tau_m \\
  \text{Read-only context}
  \bnfis& \ \Gamma \\
  \text{Read-write context}
  \bnfis& \ \Gamma; \Delta
  \end{align*}
  \end{mdframed}
  \caption{Abstract syntax}
  \label{fig:syntax}
\end{figure}

\begin{figure}[htbp]
  \centering
  \begin{mdframed}
  \begin{mathpar}
    \inferenceRule{Ty-Rw-Ro}{
      \Gamma ; \emptyctx \rwtypes c : \tau
    }{
      \Gamma \rotypes c : \tau
    }

    \inferenceRule{Ty-Ro-Rw}{
      \Gamma, \Delta \rotypes e : \tau
    }{
      \Gamma ; \Delta \rwtypes e : \tau
    }
    
    \inferenceRule{Ty-Var}{
      \Gamma(x) = \tau
    }{
      \Gamma \rotypes x : \tau
    }

    \inferenceRule{Ty-False}{
    }{
      \Gamma \rotypes \cfalse : \dB
    }

    \inferenceRule{Ty-True}{
    }{
      \Gamma \rotypes \ctrue : \dB
    }

    \inferenceRule{Ty-Int}{
    }{
      \Gamma \rotypes \numeral{k} : \dZ
    }

    \inferenceRule{Ty-Skip}{
     }{
      \Gamma \rotypes \cskip : \dU
    }

    \inferenceRule{Ty-Coerce}{
      \Gamma \rotypes e : \dZ
    }{
      \Gamma \rotypes \ccoerce{e} : \dR
    }

    \inferenceRule{Ty-Exp}{
      \Gamma \rotypes e : \dZ
    }{
      \Gamma \rotypes 2^e : \dR
    }

    \inferenceRule{Ty-Int-Op}{
      \Gamma \rotypes e_1 : \dZ \\
      \Gamma \rotypes e_2 : \dZ
    }{
      \Gamma \rotypes e_1 \iop e_2 : \dZ
    }

    \inferenceRule{Ty-Real-Op}{
      \Gamma \rotypes e_1 : \dR \\
      \Gamma \rotypes e_2 : \dR
    }{
      \Gamma \rotypes e_1 \rop e_2 : \dR
    }

    \inferenceRule{Ty-Recip}{
      \Gamma \rotypes e : \dR
    }{
      \Gamma \rotypes \cinv{e} : \dR
    }

    \inferenceRule{Ty-Int-Lt}{
      \Gamma \rotypes e_1 : \dZ \\
      \Gamma \rotypes e_2 : \dZ
    }{
      \Gamma \rotypes e_1 < e_2 : \dB
    }

    \inferenceRule{Ty-Int-Eq}{
      \Gamma \rotypes e_1 : \dZ \\
      \Gamma \rotypes e_2 : \dZ
    }{
      \Gamma \rotypes e_1 = e_2 : \dB
    }

    \inferenceRule{Ty-Real-Lt}{
      \Gamma \rotypes e_1 : \dR \\
      \Gamma \rotypes e_2 : \dR
    }{
      \Gamma \rotypes e_1 \rlt e_2 : \dB
    }

    \inferenceRule{Ty-Lim}{
      \Gamma,x \of \dZ \rotypes e : \dR
    }{
      \Gamma \rotypes (\clim{x}{e}) : \dR
    }

    \inferenceRule{Ty-Sequence}{
      \Gamma ; \Delta \rwtypes c_1 : \dU \\
      \Gamma ; \Delta \rwtypes c_2 : \tau
     }{
      \Gamma ; \Delta \rwtypes (c_1 ; c_2) : \tau
    }

    \inferenceRule{Ty-New-Var}{
      \Gamma, \Delta \rotypes e : \sigma \\
      \Gamma ; \Delta, x \of \sigma \rwtypes c : \tau
     }{
      \Gamma ; \Delta \rwtypes (\cvar{x}{e} c) : \tau
    }

    \inferenceRule{Ty-Assign}{
      \Delta \rotypes x : \tau \\
      \Gamma, \Delta \rotypes e : \tau
     }{
      \Gamma ; \Delta \rwtypes (\clet{x}{e}) : \dU
    }

    \inferenceRule{Ty-Cond}{
      \Gamma, \Delta \rotypes e : \dB \\
      \Gamma ; \Delta \rwtypes c_1 : \tau \\
      \Gamma ; \Delta \rwtypes c_2 : \tau
    }{
      \Gamma ; \Delta \rwtypes (\cif e \cthen c_1 \celse c_2 \cend) : \tau
    }

    \inferenceRule{Ty-Case}{
      \Gamma, \Delta \rotypes e_i : \dB \\
      \Gamma ; \Delta \rwtypes c_i : \tau \\
      (i = 1, \ldots, n)
    }{
      \Gamma ; \Delta \rwtypes (\ccase e_1 \To c_1 \mid \cdots \mid e_n \To c_n \cend) : \tau
    }

    \inferenceRule{Ty-While}{
      \Gamma, \Delta \rotypes e : \dB \\
      \Gamma ; \Delta \rwtypes c : \dU
    }{
      \Gamma ; \Delta \rwtypes (\cwhile e \cdo c \cend) : \dU
    }
  \end{mathpar}
  \end{mdframed}
  \caption{Typing rules}
  \label{fig:typing-rules}
\end{figure}

The abstract syntax of expressions is shown in \cref{fig:syntax}, where we indicate expressions with~$e$ and $c$ according to whether they are intended to be pure or impure. However, as discussed above, the distinction between purity and impurity is in actuality implemented by the 
type system, and the two kinds of expression share the same grammar for their abstract syntax. 

Among the pure expressions are variables~$x$, of which we presume to have an unbounded supply, boolean constants $\cfalse$ and $\ctrue$, integer numerals~$\numeral{k}$ for $k \in \IZ$, the trivial expression~$\cskip$, coercion from integers to reals $\ccoerce{e}$, exponentiation $2^e$, arithmetical operations, integer comparisons $e_1 = e_2$ and $e_1 < e_2$, real comparison $e_1 \rlt e_2$ (but not equality of reals), and the limit expression $\clim{x}{e}$.
The (potentially) state-changing expressions are sequencing $c_1; c_2$, introduction of a local variable $\cvar{x}{e} c$, variable assignment $\clet{x}{e}$, the conditional $\cif e \cthen c_1 \celse c_2 \cend$, guarded case $\ccase e_1 \To c_1 \mid \cdots \mid e_n \To c_n \cend$, and the loop $\cwhile e \cdo c \cend$. These syntactic constructions can also give rise to pure expressions  in certain circumstances dictated by the type system, which is presented in  \cref{fig:typing-rules}. 

Since the focus of Clerical is real-number computation, we briefly discuss some relevant aspects of the language and give an example program. A fully  detailed semantic explanation of the language follows in~\cref{sec:denotation}, and many more  examples appear in \cref{sec:boolean-ops,sec:example}.

The abstract datatype $\dR$ of real numbers, includes constants $\ccoerce{m}$ for every integer $m$  coerced to a real number via
the inclusion $\mathbb{Z} \subseteq \mathbb{R}$, and $2^m$ for an integer-exponent power of $2$, which of course includes fractional values when $m$ is negative. The latter construct is particularly useful for implementing bounds related to the limit operation $\clim{x}{e}$,
which assumes that $e$ defines a \emph{rapidly converging} Cauchy sequence in its integer variable $x$: namely
$(e_m)_{m \in \mathbb{Z}}$ is a Cauchy sequence whose limit $l = \lim_{m \to \infty} e_m$ satisfies
$|\,l-e_m| < 2^{-m}$, for all $m\in \mathbb{Z}$,
where we write $e_m$ for the real-number
computed by $e$ when $x$ takes value $m$. (Note that, for simplicity, and with no loss of generality, our sequences are indexed by all integers including negative indices.)
As built-in arithmetic operations on the reals, we include addition, subtraction, multiplication and reciprocal. The strict comparison operator
$e_1 \rlt e_2$ is boolean-valued, and returns the expected truth value whenever $e_1$ and $e_2$ compute to distinct reals. 
In the case of equality, however, it diverges. As discussed in Section~\ref{sec:overview}, such behaviour is an unavoidable feature of
exact-real-number computation. We do not include an equality test on reals since, for essentially the same reason, the best one could achieve is divergence in the case of equality, meaning an equality test would never return true. 

The behaviour of the guarded case construct and its relationship to such non-termination properties has been discussed in Section~\ref{sec:overview}. To end this section, we present  a short example program that demonstrates how this case construct interacts with the limit operator to define mathematically useful operations; albeit, in this case, a particularly simple one. The program below is a pure expression of type $\dR$, with a single read-only variable $x$ also of type $\dR$.  It calculates the absolute value of the real number assigned to the variable $x$. Because of the non-termination properties of exact-real-number comparison, this cannot be defined using a simple if-then-else conditional that tests whether the value of $x$ is negative or not. Instead, we need to combine the limit operator with the guarded case construct:
\begin{lstlisting}
lim n.
 case
    $x \rlt 2^{\iminus n\iminus  \numeral{1}}$ =>  $~\rminus x$ 
  | $\rminus 2^{\iminus n\iminus \numeral{1}} \rlt x$ => $~x$
 end
\end{lstlisting}
Here $\rminus e$ is an abbreviation for $\ccoerce{\numeral{0}} \rminus e$.
Given a real number $x$, to approximate its absolute value $|x|$ up to accuracy $2^{-n}$, for any integer $n$, we make a parallel test if $x < 2^{-n-1}$ holds or $x > -2^{-n-1}$ holds;  at least one of which is guaranteed to evaluate to true. 
If the first condition does, we return $-x$ as the $2^{-n}$ approximating value.
If the second condition does, we return $x$. 
In both cases, the returned value lies within $2^{-n}$ of $|x|$. Thus 
the sequence, as $n$ tends to $\infty$, has the required rapid rate of convergence to the limit value  $|x|$.

One slightly subtle point with the above example, is that the `sequence' defined by the case statement is actually nondeterministic. 
In the case that $-2^{-n-1} < x < 2^{-n-1}$, either of the values $x$ or $-x$ may arise as the $2^{-n}$ approximation. In the semantics of Clerical, we do not try to determinize this. In principle, two successive evaluations of the  $2^{-n}$ approximation, may yield two different answers. In any case, the limit of the sequence is uniquely determined. 
In the next section, 
the semantics of the limit operator defines it in general for nondeterministic sequences.

\subsection{First-order functions}
\label{sec:first-order-func}

In this section, we extend the core Clerical language presented above with first-order functions. This extension does not change the expressive power of Clerical, since  all uses functions in programs can be eliminated by inlining their definitions. Nevertheless, functions play an essential role in practical programming, since they allow  compositional programming, providing a convenient way of breaking up code into smaller, reusable pieces.

In the extended Clerical, a function is simply an abbreviation of an expression that depends on parameters.
Functions are defined using the syntax
\begin{equation*}
\cfunction f {x_1 \of \tau_1, \ldots, x_n \of \tau_n} {\sigma} {e}.
\end{equation*}
Informally, $f$ takes $n$ parameters of types $\tau_1, \ldots, \tau_n$ are returns a value of type~$\sigma$, computed by the function body~$e$.
It is a \defemph{first-order function} because the types $\tau_1, \ldots, \tau_n$ and $\sigma$ are just the primitive Clerical types, as opposed to function types (that do not exist in Clerical anyhow).

Functions are used through \defemph{function call} expressions:
\begin{equation*}
  \text{Expression}\ e, c
  \bnfis \cdots
  \bnfor f(e_1, \ldots, e_n)
  \bnfor \cdots
\end{equation*}
To keep track of the defined functions, we introduce a \defemph{top-level environment}
\begin{align*}
  T =
    [&\cfunction {f_1} {x_1 \of \tau_{1,1}, \ldots, x_{n_1} \of \tau_{1,n_1}} {\sigma_1} {e_1}), \\
     &\qquad\vdots\\
     &\cfunction {f_k} {x_1 \of \tau_{k,1}, \ldots, x_{n_1} \of \tau_{k,n_k}} {\sigma_k} {e_k})],
\end{align*}
which is just a list of function definitions.
Typing judgements are extended with~$T$ (it would suffice to keep just the signatures of the defined functions, without their bodies):
\begin{equation*}
  T ; \Gamma \rotypes e : \tau
  \qquad\text{and}\qquad
  T ; \Gamma ; \Delta \rwtypes e : \tau
\end{equation*}
The previously given inference rules just pass~$T$ from conclusions to premises. There is one new inference rule governing function calls:
\begin{equation*}
  \inferenceRule{Ty-Ro-Call}{
    (\cfunction f {x_1 \of \tau_1, \ldots, x_n \of \tau_n} {\sigma} {e}) \in T \\\\
    T ; \Gamma \rotypes e_i : \tau_i \quad \text{for $i = 1, \ldots, n$}
  }{
    T ; \Gamma \rotypes f(e_1, \ldots, e_n) : \sigma
  }
\end{equation*}
Note that function calls are pure expressions.

We also need a custom judgement for checking top-level environments:
\begin{equation*}
  \inferenceRule{Env-Empty}{ }{[\,] \isEnv}
  \qquad\qquad
  \inferenceRule{Env-Extend}{
    T \isEnv
    \\
    f \not\in T
    \\
    T ; x_1 \of \tau_1, \ldots, x_n \of \tau_n \rotypes e : \sigma
  }{
    [T, (\cfunction f {x_1 \of \tau_1, \ldots, x_n \of \tau_n} {\sigma} {e})] \isEnv
  }
\end{equation*}
The rule \rref{Env-Empty} validates the empty environment, while \rref{Env-Extend} extends an environment with a function definition, so long as the function body is a pure expression of return type when given read-only access to the arguments.
Note that each function may call previously defined functions, but may not call itself. We discuss the lack of recursion in~\cref{sec:future-work}.

\section{Denotational semantics}
\label{sec:denotation}

In this section we assign mathematical meaning to the constituent parts of Clerical.
The types are interpreted by  the expected sets:
\begin{align*}
\sem{\dZ} &\defeq \IZ &
\sem{\dB} &\defeq \{\semff, \semtt\} &
\sem{\dR} &\defeq \IR &
\sem{\dU} &\defeq \{\star\} \enspace .
\end{align*}
Typing contexts are interpreted by cartesian product:
\begin{align*}
  \sem{x_1 \of \tau_1, \ldots, x_m \of \tau_m} &\defeq
  \sem{\tau_1} \times \cdots \times \sem{\tau_m} \enspace .
\end{align*}
The denotation $\sem{\Gamma}$ of a read-only context $\Gamma$ is thought of as an \defemph{environment} specifying values of variables, whereas the denotation of a read-write context $\Gamma ; \Delta$ has two components, the environment $\sem{\Gamma}$ and the \defemph{state} $\sem{\Delta}$.

The meaning of a well-typed pure expression $\Gamma \rotypes e : \tau$ and
general expression $\Gamma; \Delta \rwtypes c : \tau$ 
 will be maps:
\begin{align*}
  \sem{\Gamma  \rotypes e : \tau}  & : \sem{\Gamma} \to \PP{\sem{\tau}}, \\
  \sem{\Gamma; \Delta \rwtypes c : \tau}  & : \sem{\Gamma} \to (\sem{\Delta} \to \PP{\sem{\Delta} \times \sem{\tau}}),
\end{align*}
where $\PP{S}$ is  a \defemph{powerdomain}, a collection of sets representing the possible sets of outcomes of nondeterministic computations that return values from $S$. Due to the distinction between error and non-termination, already mentioned in Section~\ref{sec:overview}, 
Clerical requires a rather specific form of powerdomain. In order to motivate it, we discuss the distinction between
error and non-termination in more detail.

Mathematically, we would like to  consider a well-behaved (deterministic) expression $e$ of type $\dR$ as defining a real number value $r$. Computationally, however, the best that $e$ will be able to do is to produce, on demand, an approximation of $r$ to within any specified precision. In the ideal case, given precision $\epsilon > 0$, the evaluation of $e$ will determine in finite time some rational approximation $q$ such that $|q - r| < \epsilon$. In Clerical, not every  deterministic expression of type $\dR$ achieves this ideal.

Some expressions simply give rise to non-terminating computations that never provide any approximating information. Others, may appear to provide approximating information, but do so in a way that is either incomplete or inconsistent.\footnote{Incompleteness may  arise, for example, if approximating values are only computed for $\epsilon$ that are not too small. Similarly, inconsistency can occur if two different $\epsilon_1, \epsilon_2$ result in putative approximations $q_1,q_2$ with $|q_1 - q_2| \geq \epsilon_1 + \epsilon_2$.} Our semantics of Clerical distinguishes between these two eventualities. Expressions that produce incomplete or inconsistent approximating information are considered \defemph{erroneous}, and the semantics for Clerical will ensure that no such expression is ever executed within a program with valid semantics. The motivation for this is to avoid any situation in which faulty approximations can provide misleading information. In contrast, \defemph{non-terminating} expressions are considered harmless in the sense that they cannot be a source of incorrect information. As is standard in denotational semantics, such expressions are assigned the special denotation $\bot$. 
These need to be distinguished from erroneous ones, because, as discussed in Section~\ref{sec:overview},  non-terminating expressions 
have an essential role to play  when programming in Clerical.

For any set~$S$,  define $\liftnoerr{S} \defeq S + \{\bot\}$, where~$\bot$ represents non-termination, as discussed above. 
Although $S + \{\bot\}$ is strictly speaking a coproduct (sum) of two sets, in practice we shall only use it
in instances in which $\bot \notin S$. This allows us to represent $\liftnoerr{S}$ as the (disjoint) union 
$S \cup \{\bot\}$.
Define:
\begin{equation*}
  \PP{S} \defeq
  \{ X \subseteq \liftnoerr{S} \such
       \text{$X$ infinite} \Rightarrow \bot \in X
  \},
\end{equation*}
A set $X \in  \PP{S}$ represents the results of a nondeterministic Clerical computation in the following way. Firstly, $X = \emptyset$ in the case that the computation is \emph{erroneous} (see the discussion above), in which case no result value is relevant. 
The case $\bot \in X$ applies if the nondeterministic computation has at least one non-terminating branch, in which case $X \setminus \{\bot\}$ is the set of all result values returned by terminating nondeterministic branches. If instead $\bot \notin X$ then $X$ represents the set of possible results of a necessarily terminating nondeterministic computation. Since Clerical has only  finite nondeterministic branching, such a set is necessarily finite.

We shall implicitly make use of  the fact that powerdomain $\Pstar$ carries the structure of a monad on the category of sets. For the purposes of this paper, we never need to directly refer to the abstract structure of the monad. However, the following maps associated with this structure will be useful. Firstly, for any $S$, there is a map $x \mapsto \pure{x} : S \to \PP{S}$, that maps any $x \in S$ to the singleton $\pure{x}$ representing the deterministic computation that returns $x$ as its result. 
Secondly, for any function $f : S \to  \PP{T}$ we define $\lift{f} : \PP{S} \to  \PP{T}$ by:
\[
\lift{f}(X) = 
\begin{cases} 
\emptyset & \text{if $\some{x \in X}\ f(x)= \emptyset$,}\\
\textstyle
\{\bot \such \bot \in X\} \cup \bigcup_{x \in X \setminus \{\bot\}} f(x) & \text{otherwise.}
\end{cases}
\]
It is easily checked that this indeed defines a set in $\PP{T}$. The idea behind the definition is that 
$\lift{f}(X)$ models a sequencing of nondeterministic computations: first execute the nondeterministic computation whose result is represented by $X$, then for each potential value $x \in X$
run the nondeterministic computation modelled by $f(x)$ to obtain potential return values.
This idea motivates the alternative notation
\[
\PPlet{x}{X}\, f(x)
\]
for  $\lift{f}(X)$, which we shall often use.

The reason behind the first clause in the definition of 
$\lift{f}(X)$ is that a computation is considered illegitimate if an error occurs along any possible nondeterministic branch. In such a case the entire computation is given the error denotation $\emptyset$.

\begin{figure}
  \begin{mdframed}
  \centering
  \small
\begin{align*}
  \sem{\Gamma \rotypes c : \tau} \, \gamma = {}&
     \PPlet
        {(v, ())}
        {\sem{\Gamma; \emptyctx \rwtypes c : \tau} \, \gamma \, ()}
        {\pure{v}}
  \\
  \sem{x_1 \of \tau_1; \ldots, x_n \of \tau_n \rotypes x_i : \tau_i} \, \gamma
  = {}& \pure{\gamma_i}
  \\
  \sem{\Gamma \rotypes \cfalse : \dB} \, \gamma = {}& \pure{\semff} \\
  \sem{\Gamma \rotypes \ctrue : \dB} \, \gamma = {}& \pure{\semtt} \\
  \sem{\Gamma \rotypes \numeral{k} : \dZ} \, \gamma = {}& \pure{k} \\
  \sem{\Gamma \rotypes \cskip : \dU} \, \gamma = {}& \pure{\semuu} \\
  \sem{\Gamma \rotypes \ccoerce{e} : \dR} \, \gamma = {}& \PPlet
        {z}
        {\sem{\Gamma \rotypes e : \dZ}\,\gamma}
        {\pure{\inclZ{z}}}
  \\
  \sem{\Gamma \rotypes e_1 \iop e_2 : \dZ} \, \gamma = {}&
    \PPlet
    {x}
    {\sem{\Gamma \rotypes e_1 : \dZ}\,\gamma} 
 \\
 & \PPlet
    {y}
    {\sem{\Gamma \rotypes e_2 : \dZ}\,\gamma}
    \pure{x \op y}
 \\
 \sem{\Gamma \rotypes e_1 \rop e_2 : \dR} \, \gamma = {}& 
   \PPlet
    {x}
    {\sem{\Gamma \rotypes e_1 : \dR}\,\gamma}
\\
& \PPlet
    {y}
    {\sem{\Gamma \rotypes e_2 : \dR}\,\gamma}
    \pure{x \op y}
 \\
 \sem{\Gamma \rotypes \cinv{e} : \dR} \, \gamma = {}&
 \PPlet
    {x}
    {\sem{\Gamma \rotypes e : \dR}\,\gamma}
    \begin{cases} 
    \pure{x^{-1}} & \text{if $x \neq 0$} \\
    \PPbot  & \text{if $x=0$}
    \end{cases}
 \\
  \sem{\Gamma \rotypes e_1 = e_2 : \dZ} \, \gamma = {}&
    \PPlet
    {x}
    {\sem{\Gamma \rotypes e_1 : \dZ}\,\gamma}
\\
& 
\PPlet
    {y}
    {\sem{\Gamma \rotypes e_2 : \dZ}\,\gamma}
    \begin{cases} 
    \pure{\semtt} & \text{if $x = y$} \\
    \pure{\semff} & \text{if $x \neq y$}
    \end{cases}
\\
  \sem{\Gamma \rotypes e_1 < e_2 : \dZ} \, \gamma = {}&
    \PPlet
    {x}
    {\sem{\Gamma \rotypes e_1 : \dZ}\,\gamma}
\\
& \PPlet
    {y}
    {\sem{\Gamma \rotypes e_2 : \dZ}\,\gamma}
    \begin{cases} 
    \pure{\semtt} & \text{if $x < y$} \\
    \pure{\semff} & \text{if $x \geq y$}
    \end{cases}
\\
   \sem{\Gamma \rotypes e_1 \rlt e_2 : \dR} \, \gamma = {}&
    \PPlet
    {x}
    {\sem{\Gamma \rotypes e_1 : \dR}\,\gamma} \\
    &\PPlet
    {y}
    {\sem{\Gamma \rotypes e_2 : \dR}\,\gamma}
    \begin{cases}
    \pure{\semtt} & \text{if $x < y$} \\
    \pure{\semff} & \text{if $x > y$} \\
    \PPbot &  \text{if $x = y$} 
    \end{cases}
  \\
  \sem{\Gamma \rotypes (\clim{x}{e}) : \dR}\,\gamma = {}&
  \begin{cases}
    \pure{t} &
      \begin{aligned}[t]
      &\text{if $t \in \IR$ and $\all{k \in \IZ}$} \\
      &\text{  $\sem{\Gamma, x \of \dZ \rotypes e : \dR} (\gamma, k) \subseteq \IR$ and} \\
      &\text{  $\all{u \in \sem{\Gamma, x \of \dZ \rotypes e : \dR} (\gamma, k)} |u - t| < 2^{-k}$,}
      \end{aligned}
    \\
    \PPerr   & \text{if no such $t \in \IR$ exists.}
  \end{cases}
\end{align*}
\end{mdframed}
\caption{Denotational semantics of pure expressions}
\label{figure:ro-denotations}
\end{figure}

\begin{figure}
  \begin{mdframed}
  \centering
  \small
\begin{align*}
\sem{\Gamma; \Delta \rwtypes e : \tau} \, \gamma \, \delta &=
  \begin{aligned}[t]
   &\PPlet
     {v}
     {\sem{\Gamma, \Delta \rotypes e : \tau} \, (\gamma, \delta)} \\
             &
        \quad
     \pure{(\delta, v)}
    \end{aligned}
\\
\sem{\Gamma; \Delta \rwtypes c_1 ; c_2 : \tau} \, \gamma \, \delta &=
   \begin{aligned}[t]
   &\PPlet
     {(\delta', \semuu)}
     {\sem{\Gamma; \Delta \rwtypes c_1 : \dU} \, \gamma \, \delta} 
   \\
   &  
     \quad\sem{\Gamma; \Delta \rwtypes c_2 : \tau} \, \gamma \, \delta'
       \end{aligned}
\\
\sem{\Gamma ; \Delta \rwtypes (\cvar{x}{e} c) : \tau} \, \gamma \, \delta &=
  \begin{aligned}[t]
   &\PPlet
     {v}
     {\sem{\Gamma, \Delta \rotypes e : \sigma} \, (\gamma, \delta)}
   \\
   &
    \PPletx
     {((\delta', v'), v'')}
     \\
     &
     \qquad\PPinx{\sem{\Gamma; \Delta, x \of \sigma \rwtypes c : \tau} \, \gamma \, (\delta,v)}
   \\
    &
     \quad\pure{(\delta', v'')}
  \end{aligned}
\\
\sem{\Gamma ; \Delta \rwtypes (\clet{x}{e}) : \dU} \, \gamma \, \delta &=
  \begin{aligned}[t]
  &\PPlet
    {v}
    {\sem{\Gamma, \Delta \rotypes e : \tau} \, (\gamma, \delta)}
        \\
        &
        \quad
    \pure{(\delta[x :=v]), \semuu)}
    \end{aligned}
\\
\sem{\Gamma ; \Delta \rwtypes (\cif e \cthen c_1 \celse c_2 \cend) : \tau} \, \gamma \, \delta &=
  \begin{aligned}[t]
  &\PPlet
    {b}
    {\sem{\Gamma, \Delta \rotypes e : \dU} \, (\gamma, \delta)}
    \\
    & \quad \begin{cases}
    \sem{\Gamma; \Delta \rwtypes c_1 : \tau} \, \gamma \, \delta & \text{if $b = \semtt$} \\
     \sem{\Gamma; \Delta \rwtypes c_2 : \tau} \, \gamma \, \delta & \text{if $b = \semff$}
     \end{cases}
  \end{aligned}
\end{align*}
\end{mdframed}
\caption{Denotational semantics of general expressions, excluding $\mathtt{case}$ and $\mathtt{while}$}
\label{figure:rw-denotations}
\end{figure}

\Cref{figure:ro-denotations} assigns denotational semantics to pure expressions.
One point that deserves explanation is the denotation of $\cinv{e}$, when~$e$ is a real expression representing~$0$. Since there is no appropriate real value to be given, the denotation could be chosen to be 
either $\PPerr$ or $\PPbot$. We choose the latter, as it reflects the fact that an algorithm for calculating reciprocal will run forever, given a representation of the real number~$0$, without ever returning any erroneous approximation to a result value.
Similarly, $\PPbot$ is given as the denotation of  $e_1 \rlt e_2$ when $e_1,e_2$ are two real expressions representing equal numbers, reflecting the fact that an algorithm trying to distinguish between the two numbers will run forever when given equal inputs.

The most complex definition in \cref{figure:ro-denotations} is the semantics of the  limit operation $\clim{x}{e}$.
In this definition, note that there is at most one $t \in \IR$ satisfying the first condition.
Its existence places strong requirements on the expression $e$, which must represent a sequence of sets of real numbers, such that every real number~$u$ in the $k$-th set lies within a distance of $2^{-k}$ of~$t$.
That is, every choice of a real number from every set furnishes a Cauchy sequence rapidly converging to a common limit~$t$.
The use of a sequence of sets allows the behaviour of~$e$ to be nondeterministic, but this nondeterminism is highly constrained by the common limit requirement. Furthermore, $e$ is neither allowed to diverge nor be erroneous. If any of the conditions required for~$t$ to exist fails, then the $\clim{x}{e}$ computation is declared erroneous.
This is appropriate because, in an algorithmic implementation of the limit operation, the source of error may occur deep in the computation (e.g., only at some high value for the integer~$k$) meaning that the algorithm may, before the error transpires, return erroneous information in the form of approximating values to a non-existent limit.

\Cref{figure:rw-denotations} assigns semantics to
several of the general expression  constructors: sequencing $c_1 ; c_2$,
local variable declarations $\cvar{x}{e} c$,
assignments $\clet{x_{i}}{e}$,
conditionals $\cif e \cthen c_1 \celse c_2 \cend$, and expressions $e$ qua commands.

Next, let us define the semantics of guarded choice
\begin{equation*}
  \Gamma ; \Delta \rwtypes (\ccase e_1 \To c_1 \such \cdots \such e_n \To c_n \cend) : \tau.
\end{equation*}
Using the abbreviations
$\sem{e_i} = \sem{\Gamma, \Delta \rotypes e_i : \dB}$
and
$\sem{c_i} = \sem{\Gamma; \Delta \rwtypes c_i : \tau}$
we set
\begin{align*}
  &\sem{\Gamma ; \Delta \rwtypes (\ccase e_1 \To c_1 \such \cdots \such e_n \To c_n \cend) : \tau} \, \gamma \, \delta = 
  \\
  &\quad
  \begin{cases}
    \emptyset \qquad\text{if $\some{i} \sem{e_i} \, (\gamma, \delta) = \emptyset \lor (\semtt \in \sem{e_i} \, (\gamma, \delta) \land \sem{c_i} \, \gamma \, \delta = \emptyset)$}
    \\
    S
    \cup \{ \bot \such \all{i} \sem{e_i} (\gamma, \delta) \neq \pure{\semtt} \}
    \qquad\text{otherwise}
  \end{cases}
  \\
  &\quad\text{where $S = \bigcup \left\{\sem{c_i} \, \gamma \, \delta \such
                 1 \leq i \leq n \land
                \semtt \in \sem{e_i} (\gamma, \delta) \right\}.$}
\end{align*}
The idea behind this definition is as follows. 
The $n$ (potentially nondeterministic) guard expressions $e_1, \dots, e_n$ are evaluated in parallel. 
Ignoring, for the moment, the possibility that one of these expressions might be erroneous, suppose
that one of them, 
$e_i$ say, evaluates to $\semtt$. If this occurs, then the parallel evaluation of the other guards is terminated and the continuation
$c_i$ is executed. Note that the choice of $i$ here is potentially nondeterministic.
If instead none of the guards evaluates to  $\semtt$, then none of the continuations is triggered and we consider this as amounting to nontermination (no error is caused), so we include 
$\bot$ in the denotation of the case expression. Note that this $\bot$ can arise as a result of \emph{bona fide} nontermination
(none of the guards terminates) or of deadlock (all guards terminate with $\semff$).
In the case that any of the $e_i$ guards causes an error (i.e., has denotation $\emptyset$), or if some $e_i$ has a possible evaluation to $\semtt$ that triggers the execution of an erroneous continuation command $c_i$, then the case statement is itself given the error denotation $\PPerr$.
Some of the subtleties of guarded choice are illustrated through the examples in
Section~\ref{sec:boolean-ops} below.

It remains to define the semantics of while loops.
As usual, the meaning of a while loop is required to be invariant under unrolling; i.e.,
\begin{align*}
& \sem{\Gamma;\Delta \rwtypes \cwhile e \cdo c \cend :\dU}\,\gamma\,\delta  = 
\\
& \qquad \sem{\Gamma;\Delta \rwtypes\cif e \cthen (c\, ; \, \cwhile e \cdo c \cend) \celse \cskip \cend :\dU}\,\gamma\,\delta
\end{align*}
That is, we want the value $\sem{\Gamma ; \Delta \rwtypes (\cwhile e \cdo c \cend) : \dU} \, \gamma $ to be a fixed point of the map 
\begin{equation}
  \label{eq:def-W}
  \begin{aligned}
    W_{\gamma} &:  \PP{\sem{\Delta} \times \{\semuu\}}^{\sem{\Delta}}
    \to \PP{\sem{\Delta} \times \{\semuu\}}^{\sem{\Delta}}
    \\
    W_\gamma&(h) \, \delta 
    \defeq \PPlet{b}{\sem{e}\,(\gamma,\delta)}
      \begin{cases}
        \lift{(h \circ \pi_1)}(\sem{\Gamma; \Delta \rwtypes c : \dU}\,\gamma\,\delta) & \text{if $b = \semtt$} \\
        \pure{(\delta, \semuu)} & \text{if $b = \semff$}
      \end{cases}
  \end{aligned}
\end{equation}
where $\pi_1$ is the projecting isomorphism from $\sem{\Delta} \times \{\semuu\}$ to $\sem{\Delta}$.

As is standard, we find the appropriate fixed point  by showing that $\PP{S}$ carries the structure of a domain ($\omega$-complete partial order with least element) and that $W_\gamma$ is an $\omega$-continuous function with respect to the induced pointwise order on the function space 
$\PP{\sem{\Delta} \times \{\semuu\}}^{\sem{\Delta}}$.
This allows the definition of $\sem{\Gamma ; \Delta \rwtypes (\cwhile e \cdo c \cend) : \dU} \, \gamma $ as the least fixed point $\mathsf{LFP}(W_\gamma)$ of $W_\gamma$:
\begin{equation}
\label{equation:lfp}
  \sem{\Gamma ; \Delta \rwtypes (\cwhile e \cdo c \cend) : \dU} \, \gamma  \defeq
  \mathsf{LFP}(W_\gamma) \enspace .
\end{equation}

The required partial order on $\PP{S}$ is that of the Plotkin powerdomain~\cite{plotkin1976powerdomain} modified to take account of our use of the error set $\PPerr$:
\begin{equation*}
  X \PPleq Y ~ \defiff~
  X = Y  ~\lor ~
  (\bot \in X \, \land\, (Y = \PPerr \lor (X \!\setminus\! \PPbot \subseteq Y))).
\end{equation*}
For $X, Y$ other than the error set $\PPerr$, the above coincides with the usual Egli-Milner order of the Plotkin powerdomain.
The positioning of~$\PPerr$ within the partial order is motivated by the following considerations.
The denotational semantics of $\cwhile e \cdo c \cend$ in environment~$\gamma$ is approximated by
applying $W_\gamma$ iteratively to the constantly-bottom function  $\delta \mapsto \PPbot$, yielding an $\omega$-chain
\begin{equation*}
  (\delta \mapsto \PPbot) \PPleq W_\gamma(\delta \mapsto \PPbot) \PPleq W_\gamma^2(\delta \mapsto \PPbot) \PPleq \dots
\end{equation*}
with each new approximation corresponding to one further level of unrolling of the loop. 
The presence of~$\bot$ in any $W_\gamma^n(\delta \mapsto \PPbot)\,\delta_0$ can indicate that
further unfoldings are needed to determine
$\sem{\Gamma ; \Delta \rwtypes (\cwhile e \cdo c \cend) : \dU} \, \gamma \, \delta_0$.
It is possible that some such further unfolding will result in an erroneous computation, at which point the denotational semantics will assume the value~$\PPerr$. For this reason it is necessary to have $X \PPleq \PPerr$, whenever $\bot \in X$.
In the case that $\PPbot \notin W_\gamma^n(\delta \mapsto \PPbot)\,\delta_0$, it holds that 
$\sem{\Gamma ; \Delta \rwtypes (\cwhile e \cdo c \cend) : \dU} \, \gamma \, \delta_0 = 
W_\gamma^n(\delta \mapsto \PPbot)\,\delta_0$, i.e., the semantics is fully determined at iteration $n$, and does not change under further iterations.
Thus nonempty sets $X$ with $\bot \notin X$ do not approximate $\PPerr$, i.e., $X \not\PPleq \PPerr$.

\begin{proposition} 
\label{prop:domain}
For any set $S$, it holds that $(\PP{S}, \PPleq)$ is an $\omega$-complete partial order with least element.
\end{proposition}

\begin{proof}
The least element is $\PPbot$.
For $\omega$-completeness, suppose that
\begin{equation*}
  X_0 \PPleq X_1 \PPleq X_2 \PPleq \dots
\end{equation*}
is an $\omega$-chain.
In the case that every $X_n$ contains $\bot$, it is easy to check that the supremum is $ \dsup{i} X_i \defeq \bigcup_i X_i$.
If instead there exists $n$ such that $\bot \notin X_n$ then necessarily $X_m = X_n$ for all $m \geq n$, so 
the supremum is $\dsup{i} X_i \defeq X_n$. (In the case that no $X_n$ is~$\PPerr$, the above coincides with the Plotkin powerdomain.)
\end{proof}

In the proof of the following proposition we use the \defemph{strict union} operation on $\PP{S}$, which models nondeterministic choice:
\[
X \uplus Y = \begin{cases}
\emptyset & \text{if $X = \emptyset$ or $Y = \emptyset$,} \\
X \cup Y & \text{otherwise.}
\end{cases}
\]

\begin{proposition} 
\label{prop:continuity}
The function $W_\gamma$ defined by \eqref{eq:def-W} is continuous with respect to the
pointwise partial order on $\PP{\sem{\Delta} \times \{\semuu\}}^{\sem{\Delta}}$.
\end{proposition}

\begin{proof}
The function $W_\gamma$ is the composition of several maps, two of which need their continuity checked.
The first one is monadic evaluation
\begin{align*}
   \PP{T}^S \times \PP{S} &\to \PP{T}  \\
  (g, X) &\mapsto \lift{g}(X)
\end{align*}
Monotonicity is straightforward.
To establish continuity in~$X$, consider an $\omega$-chain $X_0 \PPleq X_1 \PPleq X_2 \PPleq \dots$.
If every $X_n$ contains $\bot$, then 
\begin{multline*}
\textstyle
\lift{g}(\dsup{i}X_i) =
\lift{g}\left(\bigcup_{i}X_i\right) =
\{\bot\} \cup \bigcup \left\{g(x) \such x \in \left(\bigcup_i X_i\right) - \{\bot\}\right\}
\\
\textstyle
= \{\bot\} \cup \bigcup_i \{g(x) \such x \in X_i - \{\bot\}\}
= \dsup{i}  \lift{g}(X_i).
\end{multline*}
Otherwise the chain is eventually constant and~$g$ preserves its supremum because it is monotone.

To establish continuity in $g$, consider an $\omega$-chain $g_0 \PPleq g_1 \PPleq g_2 \PPleq \dots$ with respect to the pointwise order on $\PP{T}^S$.
If there are $k \in \IN$ and $x \in X$ such that $g_k(x) = \PPerr$ then $\lift{(\dsup{i} g_i)}(X) = \PPerr = \dsup{i} \lift{g_i}(X)$.
Thus it remains to prove that the sets
\begin{equation*}
  \textstyle
  \lift{(\dsup{i} g_i)}(X) =
      \{\bot \such \bot \in X\} \cup 
      \bigcup_{x \in X \setminus \{\bot\}} \dsup{i} g_i(x) \eqdef A
\end{equation*}
and
\begin{equation*}
  \textstyle
  \dsup{i} \lift{g_i}(X) =
  \dsup{i} \left(
      \{\bot \such \bot \in X\} \cup 
      \bigcup_{x \in X \setminus \{\bot\}} g_i(x)
   \right) \eqdef B.
\end{equation*}
are equal, assuming $g_i(x) \neq \PPerr$ for all $i \in \IN$ and $x \in X$.
Clearly, $\bot \in A \liff \bot \in B$.
To show that $A \setminus \{\bot\} = B \setminus \{\bot\}$, we first note that, for all $k \in \IN$, $x \in X$ and $y \neq \bot$,
\begin{equation*}
  y \in \dsup{i} g_i(x)
  \iff
  \some{i \in \IN} y \in g_i(x) \setminus \{\bot\},
\end{equation*}
thanks to the standing assumption that~$\PPerr$ does not appear in the supremum.
Now, if $y \in A \setminus \{\bot\}$ then $y \in \dsup{i} g_i(x)$ for some $x \in X \setminus \{\bot\}$,
hence $y \in g_j(x)$ for some $j \in \IN$, and so $y \in B$.
Conversely, if $y \in B \setminus \{\bot\}$ then $y \in g_j(x)$ for some $j \in \IN$ and $x \in X \setminus \{\bot\}$, hence $y \in \dsup{i} g_i(x)$, and so $y \in A$.

The other function partaking in $W_\gamma$ is
\begin{align*}
  C &: \PP{\{\semff, \semtt\}} \times \PP{S} \times \PP{S} \to \PP{S}
  \\
  C &: (X,Y,Z) \mapsto \PPlet{b}{X} 
      \begin{cases} 
        Y & \text{if $b = \semtt$} \\
        Z & \text{if $b = \semff$} 
      \end{cases}
\end{align*}
We only verify continuity in $Y$, which is done by case analysis on~$X$:
\begin{itemize}
\item If $\semtt \notin X$ then, $C$ is constant in~$Y$.
\item If $X = \pure{\semtt}$ then $C(\pure{\semtt},Y,Z) = Y$ is a projection.
\item If $X = \{\semtt, \semff\}$ then $C(Y) = Y \uplus Z$.
\item If $X = \{\semtt,\bot\}$ then $C(Y) = Y\cup \{\bot \such Y \neq \PPerr\}$.
\item If $X = \{\semtt, \semff,\bot\}$ then  $C(Y) = (Y \uplus Z) \cup \{\bot \such (Y \uplus Z) \neq \PPerr\}$.\end{itemize}
In each case, continuity in $Y$ is easy to show.
\end{proof}

It follows from \cref{prop:domain,prop:continuity} that the semantic definition of while commands in~\eqref{equation:lfp} is well-defined.

\subsection{Semantics of first-order functions}
\label{sec:semant-first-order}

We briefly address the denotational semantics of first-order functions from \cref{sec:first-order-func}.
The denotation of a function
\begin{equation*}
  \cfunction f {x_1 \of \tau_1, \ldots, x_n \of \tau_n} {\sigma} {e}
\end{equation*}
is just the denotation of its body,
\begin{equation*}
  \sem{f} \defeq
  \sem{x_1 \of \tau_1, \ldots, x_n \of \tau_n \rotypes e : \sigma} :
  \sem{\tau_1} \times \cdots \times \sem{\tau_n} \to \PP{\sem{\sigma}}.
\end{equation*}
We need to be a bit careful about the denotation of a function call $f(e_1, \ldots, e_n)$ because the arguments
$e_1, \ldots, e_n$ may diverge or yield nondeterministic values, so it matters if and when they are evaluated.
We opt for the call-by-value evaluation strategy that is most commonly seen in imperative languages.

Given sets~$S$ and~$T$, define the \defemph{monadic pairing}
$\PPtuple{{-}, {-}} : \PP{S} \times \PP{T} \to \PP{S \times T}$
by
\begin{equation*}
  \PPtuple{X, Y} \defeq (\PPlet{x}{X} \PPlet{y}{Y} \{(x, y)\}).
\end{equation*}
Note that the binding order does not matter, i.e., $(\PPlet{x}{X} \PPlet{y}{Y} \cdots) = (\PPlet{y}{Y} \PPlet{x}{X} \cdots)\})$ because $\Pstar$ is a commutative monad.
Monadic pairing is easily extended from pairs to $n$-tuples for arbitrary~$n$.

The denotation of a function call is application adorned with the monad structure:
\begin{align*}
  \sem{\Gamma \rotypes f(e_1, \ldots, e_n) : \sigma} \gamma \defeq &
  \lift{\sem{f}} \PPtuple{\sem{e_1} \gamma, \ldots, \sem{e_n} \gamma}.
\end{align*}

\section{Nondeterminism and parallelism}
\label{sec:boolean-ops}

The guarded case construct of Clerical requires parallel evaluation of the guards
leading to potential nondeterminism.  As a result, several basic operations involving nondeterminism and parallel evaluation are definable in Clerical.

A simple binary nondeterministic choice between two pure expressions
$\Gamma \rotypes e_1 : \tau$ and $\Gamma \rotypes e_2 : \tau$ is implemented
by
\[\Gamma \rotypes (\ccase \ctrue \To e_1 \such \ctrue \To e_2 \cend) : \tau \]
This has the derived semantics
\[
\sem{\Gamma \rotypes (\ccase \ctrue \To e_1 \such \ctrue \To e_2 \cend) : \tau }\, \gamma
~ = ~ \sem{\Gamma \rotypes e_1 : \tau } \, \gamma\,
\uplus \,
\sem{\Gamma \rotypes e_2  : \tau }\,\gamma
\]
using the strict union operation defined above Proposition~\ref{prop:continuity}.

It is also possible to implement McCarthy's \defemph{ambiguous choice} between 
$\Gamma \rotypes e_1 : \tau$ and $\Gamma \rotypes e_2 : \tau$, by:
\[\Gamma \rotypes (\ccase (\cvar{x}{e_1} \ctrue) \To e_1 \such 
(\cvar{x}{e_2} \ctrue) \To e_2 \cend) : \tau \]
Writing $\Gamma \rotypes e_1 \,\mathtt{amb}\, e_2 : \tau$ for the above, 
we have
\begin{multline*}
\sem{\Gamma \rotypes  e_1 \,\mathtt{amb}\, e_2 : \tau }\, \gamma = \\
  (\sem{\Gamma \rotypes e_1 : \tau } \, \gamma\, \uplus \, \sem{\Gamma \rotypes e_2  : \tau }\,\gamma)
  \setminus
  \{ \bot \mid \bot \notin (\sem{\Gamma \rotypes e_1  : \tau }\,\gamma 
\,\cap \,\sem{\Gamma \rotypes e_2  : \tau}\,\gamma)\}\;.
\end{multline*}

A well-known issue with ambiguous choice is that it is not monotonic with respect to any reasonable powerdomain partial ordering \cite{LEVY2007221}, meaning that it does not have a domain-theoretic semantics. Indeed, such a failure of monotonicity holds with respect to 
the ordering $\PPleq$ we have defined on our powerdomain $\PP{-}$. It follows that the denotational semantics of Clerical expressions does not, in general, act monotonically in the semantics of subexpressions. We present a  simple example of this phenomen.

Consider the expression below, which is parametric in the subexpression
$\rotypes e  : \dB$:
\begin{lstlisting}
case
| e => while true do skip end
| true => skip
end
\end{lstlisting}
In the case that $\sem{e} =  \pure{\bot}$, the denotation of the whole expression is $\{\semuu\} $. If instead $\sem{e} =  \pure{\semtt}$, then the denotation of the expression is $\{\semuu, \bot\}$. This breaks monotonicity because
$\pure{\bot} \PPleq  \pure{\semtt}$ in $\PP{\{\semff, \semtt\}}$, but $\{\semuu\} \not \PPleq \{\semuu, \bot\}$  in $\PP{\{\semuu\}}$.
Similarly, considering the case in which $e$ is an erroneous expression (i.e., $\sem{e} = \PPerr$),
we have $\pure{\bot} \PPleq  \PPerr$, but $\{\semuu\} \not \PPleq \PPerr$.

Given the non-monotonicity properties illustrated above, it is perhaps surprising that it is possible to give Clerical a denotational semantics involving a domain-theoretic fixed-point argument to establish the semantics of while loops. The reason for this is that the operator $W_\gamma$, used in defining the semantics of while loops, is defined purely as a combination of conditional statements and sequencing, and does not involve the problematic non-monotonic aspects of Clerical. Indeed, as Proposition~\ref{prop:continuity} establishes,  the particular operator $W_\gamma$ is always continuous, hence \emph{a fortiori}  monotone. 

Clerical, as we have defined it, does not include any primitive operator for manipulating booleans. This does not limit expressivity since logical operations are definable using the conditional construct.
For example, negation of a boolean expression $b$ is defined by
\[
\neg b \defeq (\cif b \cthen \cfalse \celse  \ctrue  \cend)
\]
The most concise way of defining the disjunction of two boolean typed expressions $b_1$ and $b_2$ is by $\cif b_1 \cthen \ctrue \celse  b_2 \cend$. This is asymmetric: if 
$\sem{b_1} = \pure{\semtt}$ and $\sem{b_2} = \pure {\bot}$ then the disjunction has denotation $\pure{\semtt}$, whereas if
$\sem{b_1} = \pure{\bot}$ and $\sem{b_2} = \pure {\semtt}$ then 
the resulting denotation is $\pure{\bot}$. It is similarly possible to define a symmetric strict disjunction by 
\[\cif b_1 \cthen (\cif b_2 \cthen \ctrue \celse \ctrue \cend) \celse  b_2 \cend\]
More interestingly, the guarded case construct can be used to define Plotkin's \defemph{parallel or} \cite{plotkin1977lcf}, another symmetric version of  disjunction which requires parallel evaluation, by
\begin{equation*}
b_1 \cdisj b_2 \defeq
(\ccase
      b_1 \To \ctrue
\such b_2 \To \ctrue
\such \neg b_1 \To \neg b_2
\cend)
\end{equation*}
From a logical perspective, when applied to deterministic expressions, $b_1 \cdisj b_2$ computes the disjunction of $b_1$ and $b_2$ from Kleene (and Priest) 3-valued logic.

Since Clerical is a nondeterministic language, the defined logical operators all have an induced effect on nondeterministic and erroneous expressions. For example, the derived full semantics for parallel or is:
\begin{align*}
\sem{\Gamma \rotypes b_1 \cdisj b_2 : \dB }(\gamma) &= \bigcup_{\substack{v_1 \in \sem{\Gamma \rotypes b_1:\dB}(\gamma)\\v_2 \in \sem{\Gamma \rotypes b_2 : \dB}(\gamma)}}\begin{cases}
\{\semtt\}&\text{if }v_1 = \semtt \lor v_2 = \semtt,\\
\{\semff\}&\text{if }v_1 = \semff \land v_2 = \semff,\\
\{\bot\}&\text{otherwise}.
\end{cases}
\end{align*}


\section{Specification logic}
\label{sec:specification-logic}

We devise a Hoare-style specification logic for proving the correctness of Clerical expressions.
As briefly discussed in Section~\ref{sec:overview}, 
it is necessary to have both forms of correctness: partial and total.
While both types of correctness guarantee the avoidance of erroneous execution, the difference between them is that
partial correctness allows non-termination and total correctness forbids it. 
Having both forms of correctness is necessary  for the logic's correctness and expressivity, as will become clear when we present the proof rules in \cref{sec:proof-rules}.

First, we  formulate the correctness specifications we will use, which are written in the style of Hoare-style triples. For each pure expression and general expression, we define the notions of partial and total correctness, yielding four variants, as follows.

Given a context $\Gamma = (x_1 \of \tau_1, \ldots, x_n \of \tau_n)$, a pure expression $\Gamma \rotypes e : \tau$, a \defemph{precondition} $A \subseteq \sem{\Gamma}$ and a \defemph{postcondition} $B \subseteq \sem{\Gamma} \times \sem{\tau}$, we define the \defemph{partial} and \defemph{total (read-only) correctness triples} respectively as
\begin{align*}
  \rotrip[\prt] {\Gamma} {A} {e} {B}
  &\iff
  \all{\gamma \in A} 
  (\sem{e} \, \gamma \neq \emptyset 
\,\land \,
  \all{v \in \sem{e} \, \gamma}
 (v \neq \bot \lthen (\gamma, v) \in B)),
  \\
  \rotrip[\tot] {\Gamma} {A} {e} {B}
  &\iff
   \all{\gamma \in A} 
  (\sem{e} \, \gamma \neq \emptyset 
\,\land \,
  \all{v \in \sem{e} \, \gamma}
 (v \neq \bot \land(\gamma, v) \in B)).
\end{align*}
Here, $\sem{e}$ is an abbreviation for $\sem{\Gamma \rotypes e : \tau}$.
Both forms of correctness state that the error state is not reached and that any value computed by~$e$ satisfies the postcondition. Note that partial correctness does not require termination, whereas total correctness guarantees it.

Given a logical formula $\phi$ in variables $\vec{x} = (x_1, \ldots, x_n) \in \sem{\Gamma}$ and a formula $\psi$ in variables $(\vec{x}, y) \in \sem{\Gamma} \times \sem{\tau}$, the triple
\begin{equation*}
  \rotrip[\prt]
  {\Gamma} {\{\vec{x} \in \sem{\Gamma} \mid \phi\}}
  {e}
  {\{(\vec{x}, y) \in \sem{\Gamma} \times \sem{\tau} \such \psi\}}
\end{equation*}
can be written more concisely as
\begin{equation}
  \label{eq:notation-rotrip}
  \rotrip[\prt]{\Gamma} {\phi} {e} {y \of \tau \such \psi}.
\end{equation}
Read the above as: if $\phi$ holds then $e$ does not err, and if it terminates then every value~$y$ computed by~$e$ satisfies~$\psi$. By replacing the symbol $\prt$ with $\tot$ we obtain the analogous notation for total correctness: if $\phi$ holds then~$e$ does not err, it terminates and every value~$y$ computed by~$e$ satisfies~$\psi$.

The analogous triples for general expressions are more complicated because they need to take state-change into account. Let $\Gamma; \Delta \rwtypes c : \tau$ be a general expression, $A \subseteq \sem{\Gamma} \times \sem{\Delta}$ and $B \subseteq \sem{\Gamma} \times  \sem{\Delta} \times \sem{\tau}$. Then we define the \defemph{partial} and \defemph{total (read-write) correctness triples} respectively as
\begin{multline*}
  \rwtrip[\prt] {\Gamma; \Delta} {A} {c} {B} \iff {} \\
  \all{(\gamma, \delta) \in A}
  (\sem{c} \, \gamma \, \delta \neq \emptyset  \, \land\,
  \all{w \in \sem{c} \, \gamma \, \delta} 
  (w \neq \bot \lthen (\gamma, \delta_w,v_w) \in B),
\end{multline*}
and
\begin{multline*}
  \rwtrip[\tot] {\Gamma; \Delta} {A} {c} {B} \iff {} \\
  \all{(\gamma, \delta) \in A}
  (\sem{c} \, \gamma \, \delta \neq \emptyset  \, \land\,
  \all{w \in \sem{c} \, \gamma \, \delta} \,
  (w \neq \bot \land (\gamma, \delta_w,v_w) \in B), 
\end{multline*}
where $\sem{c}$ is an abbreviation for $\sem{\Gamma; \Delta \rwtypes c : \tau}$, and we use the fact that
any element $w \in \sem{c} \, \gamma \, \delta$ which is not $\bot$ is of the form $w = (\delta_w,v_w)$, so the conclusion 
$(\gamma, \delta_w,v_w) \in B$ states that, in the new read-write state $(\gamma, \delta_w)$, the return value~$v_w$ satisfies the postcondition~$B$.

Given a formula $\phi$ in variables $(\vec{x}, \vec{y}) \in \sem{\Gamma} \times \sem{\Delta}$ and a formula $\psi$ in variables $(\vec{x}, \vec{y}, z) \in \sem{\Gamma} \times \sem{\Delta} \times \sem{\tau}$,
the triple
\begin{equation*}
  \rwtrip[\prt]
  {\Gamma; \Delta} {\{(\vec{x}, \vec{y}) \in \sem{\Gamma} \times \sem{\Delta} \such \phi\}}
  {c}
  {\{(\vec{x}, \vec{y}, z) \in \sem{\Gamma} \times \sem{\Delta} \times \sem{\tau} \such \psi\}}.
\end{equation*}
can again be written more concisely as
\begin{equation}
  \label{eq:notation-rwtrip}
  \rwtrip[\prt]{\Gamma; \Delta} {\phi} {c} {z \of \tau \such \psi}.
\end{equation}
Read the above as: if $\phi$ holds then~$c$ does not err, and if it terminates then it computes a value satisfying~$\psi$ in the new state.
By replacing~$\prt$ with~$\tot$ we get analogous notation for total correctness: if $\phi$ holds then~$c$ does not err, it terminates and computes a value satisfying~$\psi$ in the new state.

We introduce one further notational convention that streamlines reasoning about expressions of the trivial type~$\dU$: if~$\psi$ is a formula in which the variable~$y$ does not appear freely, then we write $\{ \psi \}$ instead of $\{ y \of \dU \such \psi \}$.

\subsection{Proof rules}
\label{sec:proof-rules}

The notation~\eqref{eq:notation-rotrip} is general in the sense that it can be used to express any correctness triple $\rotrip[\prt]{\Gamma}{A}{e}{B}$ by taking $\phi$ to be the formula $\vec{x} \in A$ and $\psi$ to be $(\vec{x}, y) \in B$.
A similar observation holds for~\eqref{eq:notation-rwtrip}, therefore, notations~\eqref{eq:notation-rotrip} and~\eqref{eq:notation-rwtrip} may be used freely in rules about correctness triples.

Many rules come in pairs differing only in the use of symbols~$\prt$ and~$\tot$. To avoid pointless duplication in such cases, we use the symbol~$\star$ to indicate that either~$\prt$ or~$\tot$ can be inserted in its place, consistently throughout a rule. For example, the rule \rref{Ro-Conj} in \cref{fig:rules-general} decompresses to the rules
\begin{mathpar}
  \inferenceRule{Ro-Part-Conj}{
    \rotrip[\prt] {\Gamma} {\phi} {e} {y \of \tau \such \psi_1} \\
    \rotrip[\prt] {\Gamma} {\phi} {e} {y \of \tau \such \psi_2}
  }{
    \rotrip[\prt] {\Gamma} {\phi} {e} {y \of \tau \such \psi_1 \land \psi_2}
  }

  \inferenceRule{Ro-Tot-Conj}{
    \rotrip[\tot] {\Gamma} {\phi} {e} {y \of \tau \such \psi_1} \\
    \rotrip[\tot] {\Gamma} {\phi} {e} {y \of \tau \such \psi_2}
  }{
    \rotrip[\tot] {\Gamma} {\phi} {e} {y \of \tau \such \psi_1 \land \psi_2}
  }
\end{mathpar}
We partition the rules into three groups:
the \emph{general rules} in \cref{fig:rules-general},
the \emph{arithmetical rules} in \cref{fig:rules-arithmetic}, and
the \emph{operational rules} in \cref{fig:rules-operational}.
In all of them, it is presupposed that the expressions appearing are well-typed in the indicated contexts and with the indicated type and that all the formulas are well-scoped in the given contexts.
Many rules have a logical side-condition for the rule to apply, written as an additional premise. 
When a formula $\phi$ in variables $\vec{x} = (x_1, \ldots, x_n) \in \sem{\Gamma}$  is written as a premise, it means that the rule only applies in the case that $\phi$ holds for all $\vec{x} \in \sem{\Gamma}$.

\begin{figure}
  \centering
  \begin{mdframed}
    \footnotesize
    \centering

    \textbf{Logical rules:}
    \begin{mathpar}
      \inferenceRule{Ro-Imply}{
        \rotrip {\Gamma} {\phi'} {e} {y \of \tau \such \psi'} \\\\
        \phi \lthen \phi' \\
        \psi' \lthen \psi
      }{
        \rotrip {\Gamma} {\phi} {e} {y \of \tau \such \psi}
      }

      \inferenceRule{Rw-Imply}{
        \rwtrip {\Gamma; \Delta} {\phi'} {c} {z \of \tau \such \psi'} \\\\
        \phi \lthen \phi' \\
        \psi' \lthen \psi
      }{
        \rwtrip {\Gamma; \Delta} {\phi} {c} {z \of \tau \such \psi}
      } \\

      \inferenceRule{Ro-ExFalso}{
      }{
        \rotrip {\Gamma} { \bot } {e} { \psi }
      }

      \inferenceRule{Rw-ExFalso}{
      }{
        \rwtrip {\Gamma; \Delta} { \bot } {c} { \psi }
      }

      \inferenceRule{Ro-Conj}{
        \rotrip {\Gamma} {\phi} {e} {y \of \tau \such \psi_1} \\
        \rotrip {\Gamma} {\phi} {e} {y \of \tau \such \psi_2}
      }{
        \rotrip {\Gamma} {\phi} {e} {y \of \tau \such \psi_1 \land \psi_2}
      }

      \inferenceRule{Rw-Conj}{
        \rwtrip {\Gamma; \Delta} {\phi} {c} {z \of \tau \such \psi_1} \\
        \rwtrip {\Gamma; \Delta} {\phi} {c} {z \of \tau \such \psi_2}
      }{
        \rwtrip {\Gamma; \Delta} {\phi} {c} {z \of \tau \such \psi_1 \land \psi_2}
      }

      \inferenceRule{Ro-Disj}{
        \rotrip {\Gamma} {\phi_1} {e} {y \of \tau \such \psi} \\
        \rotrip {\Gamma} {\phi_2} {e} {y \of \tau \such \psi}
      }{
        \rotrip {\Gamma} {\phi_1 \lor \phi_2} {e} {y \of \tau \such \psi}
      }

      \inferenceRule{Rw-Disj}{
        \rwtrip {\Gamma; \Delta} {\phi_1} {c} {z \of \tau \such \psi} \\
        \rwtrip {\Gamma; \Delta} {\phi_2} {c} {z \of \tau \such \psi}
      }{
        \rwtrip {\Gamma; \Delta} {\phi_1 \lor \phi_2} {e} {z \of \tau \such \psi}
      }
    \end{mathpar}
    
    \medskip

    \textbf{Variables and constants:}
    \begin{mathpar}
      \inferenceRule{Ro-Var}{
      }{
        \rotrip {\Gamma_1, x \of \tau, \Gamma_2} {\psi[x/y]} {x} {y \of \tau \such \psi}
      }

      \inferenceRule{Ro-Skip}{
      }{
        \rotrip {\Gamma} {\psi} \cskip {\psi}
      }

      \inferenceRule{Ro-Bool-Const}{
        b \in \{\cfalse, \ctrue\}
      }{
        \rotrip {\Gamma} {\psi[b/y]} b {y \of \dB \such \psi}
      }

      \inferenceRule{Ro-Int-Const}{
        k \in \IZ
      }{
        \rotrip {\Gamma} {\psi[k/y]} {\numeral{k}} {y \of \dZ \such \psi}
      }
    \end{mathpar}

\medskip

    \textbf{Passage between read-only and read-write correctness:}
    \begin{mathpar}
      \inferenceRule{Rw-Ro}{
        \rwtrip {\Gamma; \emptyctx} {\phi} {c} {z \of \tau \such \psi}
      }{
        \rotrip {\Gamma} {\phi} {c} {z \of \tau \such \psi}
      }

      \inferenceRule{Ro-Rw}{
        \rotrip {\Gamma, \Delta} {\phi} {e} {y \of \tau \such \psi}
      }{
        \rwtrip {\Gamma; \Delta} {\phi} {e} {y \of \tau \such \psi}
      }
    \end{mathpar}

  \end{mdframed}
  \caption{General proof rules}
  \label{fig:rules-general}
\end{figure}

\begin{figure}
  \centering
  \begin{mdframed}
    \footnotesize
    \centering

    \textbf{Coercion and exponentiation:}
    \begin{mathpar}
      \inferenceRule{Ro-Coerce}{
        \rotrip {\Gamma} {\phi} e {y \of \dZ \such \psi[\inclZ{y}/y]}
      }{
        \rotrip {\Gamma} {\phi} {\ccoerce{e}} {y \of \dR \such \psi}
      }

      \inferenceRule{Ro-Exp}{
        \rotrip {\Gamma} {\phi} {e} {x \of \dZ \such \psi[2^x/y]} 
      }{
        \rotrip {\Gamma} {\phi} {2^e} {y \of \dR \such \psi}
      }

    \end{mathpar}
    
    \medskip

    \textbf{Integer arithmetic:}
    \begin{mathpar}
      \inferenceRule{Ro-Int-Op}{
        \rotrip {\Gamma} {\phi} {e_1} {y_1 \of \dZ \such \psi_1} \\
        \rotrip {\Gamma} {\phi} {e_2} {y_2 \of \dZ \such \psi_2} \\\\
        \psi_1 \land \psi_2 \lthen \psi[(y_1 \op y_2)/y]
      }{
        \rotrip {\Gamma} {\phi} {e_1 \iop e_2} {y \of \dZ \such \psi}
      }
    \end{mathpar}

\medskip

    \textbf{Real arithmetic:}
    \begin{mathpar}
      \inferenceRule{Ro-Real-Op}{
        \rotrip {\Gamma} {\phi} {e_1} {y_1 \of \dR \such \psi_1} \\
        \rotrip {\Gamma} {\phi} {e_2} {y_2 \of \dR \such \psi_2} \\\\
        \psi_1 \land \psi_2 \lthen \psi[(y_1 \op y_2)/y]
      }{
        \rotrip {\Gamma} {\phi} {e_1 \rop e_2} {y \of \dR \such \psi}
      }
    \end{mathpar}

\medskip

    \textbf{Reciprocal:}
    \begin{mathpar}
      \inferenceRule{Ro-Part-Recip}{
        \rotrip[\prt] {\Gamma} {\phi} {e} {x \of \dR \such \theta} \\\\
        \theta \land x \neq 0 \lthen \psi[x^{-1}/y]
      }{
        \rotrip[\prt] {\Gamma} {\phi} {\cinv{e}} {y \of \dR \such \psi}
      }

      \inferenceRule{Ro-Tot-Recip}{
        \rotrip[\tot] {\Gamma} {\phi} {e} {x \of \dR \such \theta} \\\\
        \theta \lthen x \neq 0 \land \psi[x^{-1}/y]
      }{
        \rotrip[\tot] {\Gamma} {\phi} {\cinv{e}} {y \of \dR \such \psi}
      }
    \end{mathpar}

\medskip

    \textbf{Integer comparison ${\sim} \in \{{\ieq}, {\ilt}\}$:}
    \begin{mathpar}
      \inferenceRule{Ro-Int-Comp}{
        \rotrip {\Gamma} {\phi} {e_1} {y_1 \of \dZ \such \psi_1} \\
        \rotrip {\Gamma} {\phi} {e_2} {y_2 \of \dZ \such \psi_2} \\\\
        \psi_1 \land \psi_2 \lthen \psi[(y_1 \sim y_2)/b]
      }{
        \rotrip {\Gamma} {\phi} {e_1 \sim e_2} {b \of \dB \such \psi}
      }
    \end{mathpar}

\medskip

    \textbf{Real comparison:}
    \begin{mathpar}
      \inferenceRule{Ro-Part-Real-Lt}{
        \rotrip[\prt] {\Gamma} {\phi} {e_1} {y_1 \of \dR \such \psi_1} \\
        \rotrip[\prt] {\Gamma} {\phi} {e_2} {y_2 \of \dR \such \psi_2} \\\\
        \psi_1 \land \psi_2 \land y_1 \neq y_2 \lthen \psi[(y_1 < y_2)/b]
      }{
        \rotrip[\prt] {\Gamma} {\phi} {e_1 \rlt e_2} {b \of \dB \such \psi}
      }

      \inferenceRule{Ro-Tot-Real-Lt}{
        \rotrip[\tot] {\Gamma} {\phi} {e_1} {y_1 \of \dR \such \psi_1} \\
        \rotrip[\tot] {\Gamma} {\phi} {e_2} {y_2 \of \dR \such \psi_2} \\\\
        \psi_1 \land \psi_2 \lthen y_1 \neq y_2 \land \psi[(y_1 < y_2)/b]
      }{
        \rotrip[\tot] {\Gamma} {\phi} {e_1 \rlt e_2} {b \of \dB \such \psi}
      }
    \end{mathpar}

\medskip

    \textbf{Limit:}
    \begin{mathpar}
      \inferenceRule{Ro-Lim}{
        \rotrip[t] {\Gamma, x \of \dZ} {\phi} {e} {z \of \dR \such \theta} \\
        \phi \lthen \some{y \in \IR} \psi \land \all{x \in \IZ} \all{z \in \IR} \theta \lthen |z - y| < 2^{-x}
      }{
        \rotrip {\Gamma} {\phi} {\clim{x} e} {y \of \dR \such \psi}
      }
    \end{mathpar}

  \end{mdframed}
  \caption{Arithmetical proof rules}
  \label{fig:rules-arithmetic}
\end{figure}

\begin{figure}
  \centering
  \begin{mdframed}
    \footnotesize
    \centering

    \begin{mathpar}
      \inferenceRule{Rw-Sequence}{
        \rwtrip {\Gamma; \Delta} {\phi} {c_1} {\theta} \\
        \rwtrip {\Gamma; \Delta} {\theta} {c_2} {z \of \tau \such \psi}
      }{
        \rwtrip {\Gamma; \Delta} {\phi} {c_1 ; c_2} {z \of \tau \such \psi}
      }

      \inferenceRule{Rw-New-Var}{
        \rotrip {\Gamma, \Delta} {\phi} {e} {x \of \sigma \such \theta} \\
        \rwtrip {\Gamma; \Delta, x \of \sigma} {\theta} {c} {z \of \tau \such \psi}
      }{
        \rwtrip {\Gamma; \Delta} {\phi} {\cvar{x}{e} c} {z \of \tau \such \psi}
      }

      \inferenceRule{Rw-Assign}{
        \rotrip {\Gamma, \Delta} {\phi} {e} {x \of \tau \such \theta} \\
        \theta \lthen \psi[x/y]
      }{
        \rwtrip {\Gamma; \Delta_1, y \of \sigma, \Delta_2} {\phi} {\clet{y}{e}} {\psi}
      }

      \inferenceRule{Rw-Cond}{
        \rotrip {\Gamma, \Delta} {\phi} {e} {b \of \dB \such \theta } \\
        \rwtrip {\Gamma; \Delta} {\theta[\semtt/b]} {c_1} {z \of \tau \such \psi} \\
        \rwtrip {\Gamma; \Delta} {\theta[\semff/b]} {c_2} {z \of \tau \such \psi} \\
      }{
        \rwtrip {\Gamma; \Delta} {\phi} {\cif e \cthen c_1 \celse c_2 \cend} {z \of \tau \such \psi}
      }

      \inferenceRule{Rw-Part-Case}{
        {\begin{aligned}
        &\rotrip[\prt] {\Gamma, \Delta} {\phi} {e_i} {b \of \dB \such \theta_i }
         &&\text{for $i = 1, \ldots, n$} \\[-0.25em]
        &\rwtrip[\prt] {\Gamma; \Delta} {\theta_i[\semtt/b]} {c_i} {z \of \tau \such \psi}
         &&\text{for $i = 1, \ldots, n$}
        \end{aligned}}
      }{
        \rwtrip[\prt]
          {\Gamma; \Delta}
          {\phi}
          {\ccase e_1 \To c_1 \mid \cdots \mid e_n \To c_n \cend}
          {z \of \tau \such \psi}
      }

      \inferenceRule{Rw-Tot-Case}{
        {\begin{aligned}
          &\rotrip[\prt] {\Gamma, \Delta} {\phi} {e_i} {b \of \dB \such \theta_i}
            &&\text{for $i = 1, \ldots, n$} \\[-0.25em]
          &\rwtrip[\tot] {\Gamma; \Delta} {\theta_i[\semtt/b]} {c_i} {z \of \tau \such \psi}
            &&\text{for $i = 1, \ldots, n$} \\[-0.25em]
          &\phi \lthen \phi_1 \lor \cdots \lor \phi_n \\[-0.25em]
          &\rotrip[\tot] {\Gamma, \Delta} {\phi_i} {e_i} {b \of \dB \such b = \semtt}
            &&\text{for $i = 1, \ldots, n$}
        \end{aligned}}
      }{
        \rwtrip[\tot]
          {\Gamma; \Delta}
          {\phi}
          {\ccase e_1 \To c_1 \mid \cdots \mid e_n \To c_n \cend}
          {z \of \tau \such \psi}
      }

      \inferenceRule{Rw-Part-While}{
        \rotrip[\prt] {\Gamma, \Delta} {\phi} {e} {b \of \dB \such \theta} \\
        \rwtrip[\prt] {\Gamma; \Delta} {\theta[\semtt/b]} {c} {\phi}
      }{
        \rwtrip[\prt] {\Gamma; \Delta} {\phi} {\cwhile e \cdo c \cend} {\phi \land \theta[\semff/b]}
      }

      \inferenceRule{Rw-Tot-While}{
        \rotrip[\tot] {\Gamma, \Delta} {\phi} {e} {b \of \dB \such \theta} \\
        \rwtrip[\tot]
           {\Gamma; \Delta}
           {\theta[\semtt/b]}
           {c}
           {\phi}
        \\\\
        \phi \lthen (\text{$\psi$ well-founded}) \\
        \rwtrip[\tot]
           {\vec{z} \of \Delta', \Gamma; \vec{y} \of \Delta}
           {\theta[\semtt/b] \land \vec{z} = \vec{y} }
           {c}
           {\psi}
      }{
        \rwtrip[\tot] {\Gamma; \Delta} {\phi} {\cwhile e \cdo c \cend} {\phi \land \theta[\semff/b] }
      }

    \end{mathpar}

  \end{mdframed}
  \caption{Operational proof rules}
  \label{fig:rules-operational}
\end{figure}

The general rules in \cref{fig:rules-general} postulate logical principles and govern variables and constants. The rules \rref{Ro-Rw} and \rref{Rw-Ro} allow passage between reasoning about pure and general expressions.

Among the arithmetical rules in \cref{fig:rules-arithmetic} we point out the difference between the partial reciprocal rule \rref{Ro-Part-Recip} and its total variant \rref{Ro-Tot-Recip}. In the former, we may assume the argument to be non-zero, while in the latter we have to prove it.
A similar phenomenon happens with the strict comparison~$x \rlt y$ on the reals, where the partial rule \rref{Ro-Part-Real-Lt} provides the assumption $x \neq y$, whereas its total variant \rref{Ro-Tot-Real-Lt} requires a proof of $x \neq y$. The condition $x \neq y$ appears in these rules because $x \rlt y$ diverges when $x = y$. 

The limit rule \rref{Ro-Lim} states that $\clim{x} e$ computes $y \in \IR$ when~$e$ denotes a sequence rapidly converging to~$t$, i.e., any value computed by~$e$ is within distance~$2^{-x}$ of~$y$.
Note that the first premise imposes \emph{total} correctness, even in the partial version of the rule. This is necessary because the denotation of $\clim{x} e$ is~$\PPerr$ (rather than~$\PPbot$) as soon~$e$ fails to be total.

The first three operational rules in \cref{fig:rules-operational}, which regulate sequencing, new variables and assignment, require no comment.
The remaining rules all handle the conditions and guards in a similar fashion,
so we only look at the rule for conditional \rref{Rw-Cond}, keeping in mind that the boolean expression~$e$ is nondeterministic. The postcondition~$\theta$ describes the \emph{possible} values~$b$ of~$e$. The preconditions $\theta[\semtt/b]$ and $\theta[\semff/b]$ of~$c_1$ and $c_2$ should respectively be read as ``if $e$ evaluates to true'' and ``if $e$ evaluates to false''. 
To see this formally, consider the triple in the premise
\[
\rotrip {\Gamma, \Delta} {\phi} {e} {b \of \dB \such \theta}
\]
and a state $\gamma \in\sem{\Gamma}$ that satisfies the precondition $\phi$.
If $\semtt \in \sem{\Gamma \rotypes e : \dB}\,\gamma$, by the definition of triples, $\gamma$ satisfies $\theta[\semtt/b]$.
Hence, $\theta[\semtt/b]$ is a necessary condition of the states satisfying $\phi$ to let $e$ admit a nondeterministic branch leading to $\semtt$,  allowing a possibility of having some branches fail to terminate when $\star = \prt$.
Due to nondeterminism, the conditions $\theta[\semtt/b]$ and $\theta[\semff/b]$ need not exclude each other, even when~$\theta$ is as informative as it can be.
Incidentally, if desired the original precondition $\phi$ may be incorporated into~$\theta$ using the admissible rule \rref{Ro-Refine}; see \cref{sec:admissible-rules}.

The rules for nondeterministic guarded cases employ a similar technique. The possible values of the guards~$e_i$ are described by the corresponding postconditions~$\theta_i$, after which each case is considered under the
precondition~$\theta_i[\semtt/b]$. Again, $\theta_i[\semtt/b]$ need not preclude $\theta_i[\semff/b]$, nor any of the other preconditions $\theta_j[\semtt/b]$ with $j \neq i$.
One notable point in the total variant \rref{Rw-Tot-Case} is that it is not obtained from \rref{Rw-Part-Case} by changing $\prt$ to $\tot$.
In \rref{Rw-Tot-Case}, each guard is required to be equipped with one partial correctness triple and one total correctness triple:
\[\rotrip[\prt] {\Gamma, \Delta} {\phi} {e_i} {b \of \dB \such \theta_i}
\quad\text{and}\quad
\rotrip[\tot] {\Gamma, \Delta} {\phi_i} {e_i} {b \of \dB \such b = \semtt}
\]
The triple for capturing the values of $e_i$ is allowed to stay \emph{partial}, 
though the guarded case requires to be total. 
The reason is that the total correctness of the overall case expression does not require each of the guards to be total and
we only need to make sure that there is a guard that can be chosen.
The extra implication in the premise 
$\phi \lthen \phi_1 \lor \cdots \lor \phi_n$ ensures that.
The precondition $\phi_i$ of the total correctness triple 
stands for the condition in which $e_i$ terminates and evaluates to and only to $\semtt$. 

Lastly, we explain the rules for the loop $\cwhile e \cdo c \cend$. A superficial reading of \rref{Rw-Part-While} looks suspect because~$c$ seemingly need not satisfy an invariant. One has to read both premises together:  the invariant~$\phi$ starts as the precondition for~$b$ passes to $c$ via an intermediate statement~$\theta$, and emerges as the postcondition for~$c$.

The total rule \rref{Rw-Tot-While} establishes an invariant the same way, and ensures termination by means of a well-founded relation, as follows.
The formula~$\psi$ involves the read-only variables~$\Gamma$, the mutable variables~$\vec{y} : \Delta$, and a read-only copy $\vec{z} : \Delta'$ of the mutable variables.  Given a read-only state $\gamma \in \sem{\Gamma}$, an \defemph{infinite $\psi$-chain} is a sequence $s : \IN \to \sem{\Delta}$ of states such that $\psi[s_i/\vec{z}, \gamma/\vec{x}, s_{i+1}/\vec{y}]$ holds for all $i \in \IN$. Say that~$\psi$ is \defemph{well-founded} under condition $\phi$ when, for every $\gamma \in \sem{\Gamma}$ and $\delta \in \sem{\Delta}$ satisfying $\phi[\gamma/\vec{x}, \delta/\vec{y}]$, there is no infinite $\psi$-chain starting from $\delta$.
In the rule \rref{Rw-Tot-While},
the precondition $\vec{z} = \vec{y}$ and the postcondition~$\psi$ together express the fact that the state just before the execution of~$c$ and the state just after forms a link in a $\psi$-chain. The loop must therefore terminate, or else it would yield an infinite $\psi$-chain.

\begin{theorem}\label{t:sound}
  The proof rules given in \cref{fig:rules-general,fig:rules-arithmetic,fig:rules-operational} are sound:
  a derivable correctness triple is valid.
\end{theorem}

\begin{proof}
  The proof proceeds by induction on the derivation of a triple, which amounts to checking for each rule that its conclusion is valid if the premises are.
  Establishing the soundness of general rules in \cref{fig:rules-general} is just an easy exercise in logic.

  To check the soundness of the arithmetical rules in \cref{fig:rules-arithmetic} one has to unwind the semantics of the premises and conclusions, and compare them to the denotational semantics of expressions involved.
  We spell out just the total version of the limit rule \rref{Ro-Lim}.
  Consider any $\gamma \in \sem{\Gamma}$ and $k \in \IZ$ such that $\phi(\gamma, k)$. The first premise guarantees that $\sem{e}(\gamma, k) \subseteq \IR$, as well as $\theta(\gamma, k, u)$ for all $u \in \sem{e}(\gamma, k)$.
  The second premise further ensures the existence of $t \in \IR$ such that $\psi(\gamma, k, t)$ and $\all{u \in \sem{e} (\gamma, k)} |u - t| \leq 2^{-k}$. Thus, the conditions of the first clause in the semantics of $\clim{x} e$ in \cref{figure:ro-denotations} are met, hence $\sem{\clim{x} e}\,\gamma = \{t\}$, yielding the desired conclusion.
\end{proof}

\subsection{Admissible rules}
\label{sec:admissible-rules}

Some useful admissible rules are collected in \cref{fig:rules-admissible}. There are verified by structural induction on the derivation of the premises, see \cref{sec:formalization}.
The rules \rref{Ro-Tot-Part} and \rref{Rw-Tot-Part} allow us to pass from total to partial correctness.
The rules \rref{Ro-Refine} and \rref{Rw-Refine} are useful for eliding an statement~$\theta$ that does not play any role in a proof, where the latter rule requires on the side that~$\theta$ not mention the read-write variables.

\begin{figure}[htbp]
  \centering
  \begin{mdframed}
    \footnotesize
    \centering

\begin{mathpar}
  \inferenceRule{Ro-Tot-Part}{
    \rotrip[\tot] {\Gamma} {\phi} {e} {y \of \tau \such \psi}
  }{
    \rotrip[\prt] {\Gamma} {\phi} {e} {y \of \tau \such \psi}
  }

  \inferenceRule{Rw-Tot-Part}{
    \rwtrip[\tot] {\Gamma; \Delta} {\phi} {c} {z \of \tau \such \psi}
  }{
    \rwtrip[\prt] {\Gamma; \Delta} {\phi} {c} {z \of \tau \such \psi}
  }

  \\

  \inferenceRule{Ro-Refine}{
    \rotrip {\Gamma} {\phi} {e} {y \of \tau \such \psi}
  }{
    \rotrip {\Gamma} {\phi \land \theta} {e} {y \of \tau \such \psi \land \theta}
  }

  \inferenceRule{Rw-Refine}{
    \rwtrip {\Gamma;\Delta} {\phi} {e} {y \of \tau \such \psi} \\
    \fv{\theta} \cap \dom{\Delta} = \emptyset
  }{
    \rwtrip {\Gamma;\Delta} {\phi \land \theta} {e} {y \of \tau \such \psi \land \theta}}
\end{mathpar}

  \end{mdframed}
  \caption{Admissible rules}
  \label{fig:rules-admissible}
\end{figure}

\subsection{Specification of functions and function calls}
\label{sec:spec-first-order}

Given a function definition
\begin{equation*}
  \cfunction f {x_1 \of \tau_1, \ldots, x_n \of \tau_n} {\sigma} {e}
\end{equation*}
the rule governing function calls to~$f$ is
\begin{equation*}
  \inferenceRule{Ro-Call}{
    \rotrip {\Gamma} {\phi} {e_i} {x_i \of \tau_i \such \theta_i} \quad\text{for $i = 1, \ldots, n$}
    \\\\
    \rotrip {\Gamma, x_1 \of \tau_1, \ldots, x_n \of \tau_n}
            {\theta_1 \land \cdots \land \theta_n}
            {e}
            {y \of \sigma \such \psi}
  }{
    \rotrip {\Gamma} {\phi} {f(e_1, \ldots, e_n)} {y \of \sigma \such \psi}
  }
\end{equation*}
It is typically used indirectly, as follows. Upon defining~$f$, we prove an assertion
\begin{equation}
  \label{eq:f-assertion}
  \rotrip
  {x_1 \of \tau_1, \ldots, x_n \of \tau_n}
  {\phi_f}
  {e}
  {y \of \sigma \such \psi_f}
\end{equation}
that is deemed to usefully characterize the body of the definition~$e$. Then, to establish
\begin{equation}
  \label{eq:call-target}
  \rotrip {\Gamma} {\phi} {f(e_1, \ldots, e_n)} {y \of \sigma \such \psi}
\end{equation}
we prove for each $i = 1, \ldots, n$ an assertion
$\rotrip {\Gamma} {\phi} {e_i} {x_i \of \tau_i \such \theta_i}$
such that the implication $\theta_1 \land \cdots \land \theta_n \lthen \phi_f$ holds, we verify $\psi_f \lthen \psi$, and appeal to \rref{Ro-Call} and~\eqref{eq:f-assertion} to conclude~\eqref{eq:call-target}.
The method is demonstrated in the next section.


\section{Example: computation of \texorpdfstring{$\pi$}{π}}
\label{sec:example}

To showcase Clerical and its specification logic, we construct a program that computes~$\pi$ and prove that it really does so.
We carry out the task in three steps: the definition of absolute value, the definition of~$\sin$, and a computation of~$\pi$ by a zero-finding procedure. Because the definition of $\sin$ uses absolute values, and computation of~$\pi$  uses~$\sin$, we first define 
functions $\mathtt{abs}$ and $\mathtt{sin}$, using the function notation of \cref{sec:first-order-func}.

\paragraph*{Absolute value}
\label{sec:absolute-value}

As a warm-up exercise we prove that the function $\mathtt{abs}$ defined in \cref{fig:abs-def}, whose body we already discussed in \cref{sec:syntax}, computes absolute values.
\begin{figure}[htb]
  \centering
  \begin{mdframed}
  \small
\begin{lstlisting}
let $\mathtt{abs}$(x:$\dR$):$\dR$ :=
 lim n.
   case
      $x \rlt 2^{\iminus n\iminus  \numeral{1}}$ =>  $~\rminus x$ 
    | $\rminus 2^{\iminus n\iminus \numeral{1}} \rlt x$ => $~x$
   end
\end{lstlisting}
  \end{mdframed}
  \caption{Implementation of absolute value.}
  \label{fig:abs-def}
\end{figure}

\begin{figure}[htb]
{\centering\small
\begin{mdframed}
\begin{lstlisting}
[\claim{ \top }]
lim n .
  [\claim{ \top }]
  case
    [\claim{ \top }] $x \rlt 2^{\iminus n\iminus  \numeral{1}}$ [\claimp{ b :  \dB \such b  \lthen x < 2^{-n-1} }]  =>
    [\claim{ x < 2^{-n-1}}] $\rminus x$ [\claimt{ z  :  \dR \such \abs{z - \abs{x}} < 2^{-n} }]
  | [\claim{ \top }] $\rminus 2^{\iminus n\iminus \numeral{1}} \rlt x$ [\claimp{ b :  \dB \such b \lthen - 2^{-n-1} < x}] =>
    [\claim{\rminus 2^{-n-1} < x}] x [\claimt{ z  :  \dR \such \abs{z - \abs{x}} < 2^{-n} }]
  end
  [\claimt{ z  :  \dR \such \abs{z - \abs{x}} < 2^{-n} }]
[\claimt{ y  :  \dR \such y = \abs{x} }]
\end{lstlisting}
\end{mdframed}
}
  \caption{Specification of $\mathtt{abs}$.}
  \label{fig:abs-correct}
\end{figure}
The specification and proof of correctness of the body of~$\mathtt{abs}$ is shown in \cref{fig:abs-correct}.
As is customary, we interleaved code and assertions to give a sequence of reasoning steps leading from the initial precondition to the final postcondition. That is, each line of code is preceded by a precondition and succeeded by a postcondition, which doubles as the precondition for the next line code.

The outer assertion in \cref{fig:abs-correct} states that any value $y : \dR$ computed by the limit equals~$\abs{x}$.  This is established by an application of \rref{Ro-Lim}, which creates the obligation that the any value $z : \dR$ computed by the case expression is within $2^{-n}$ of~$\abs{x}$. This in turn is proved by an applying \rref{Rw-Tot-Case}, which generates obligations not shown in \cref{fig:abs-correct}, namely:
\begin{itemize}
\item $\top \lthen x < 2^{-n - 1} \lor -2^{-n-1} < x$,
\item $\{ x < 2^{-n-1} \}\; x < 2^{-n - \numeral{1}}\; \{ b : \dB \mid b \}^\tot$,
\item $\{ -2^{-n-1} < x \}\; -2^{-n - \numeral{1}} < x \; \{ b : \dB \mid b \}^\tot$.
\end{itemize}
These obviously hold.
Finally, there is an additional obligation due to \rref{Ro-Lim},
\begin{equation*}
  \top \lthen
   \some{y \in \IR}
   y = \abs{x} \land
   \all{n \in \IZ, z\in\IR} \abs{z - \abs{x}} < 2^{-n} \lthen \abs{z - y} < 2^{-n},
\end{equation*}
which holds because we may take $y = \abs{x}$.

\paragraph*{Sine function}

We compute the sine function by employing its Taylor expansion at~$0$:
\begin{equation*}
  \sin(x) = \sum_{i = 0}^\infty \frac{(-1)^i x^{2i+1}}{(2i+1)!}.
\end{equation*}
The method is inefficient for large~$x$, but is good enough for $3 < x < 4$, which is what we need.
Define the terms of the expansion and the partial sums
\begin{equation*}
  \mathbf{t}(i, x) \defeq \frac{(-1)^i x^{2i+1}}{(2i+1)!}
  \qquad\text{and}\qquad
  \mathbf{S}(j, x) \defeq \sum_{i=0}^j \mathbf{t}(i, x),
\end{equation*}
and recall the error bound for the $j$-th partial sum
\begin{equation}
  \label{eq:exa:sine-bound}
  \abs{\sin(x) - \mathbf{S}(j,x)} \leq \abs{\mathbf{t}(j+1, x)}.
\end{equation}
The recursive relations
\begin{align*}
  \mathbf{t}(j+1, x) &= -\mathbf{t}(j, x) \cdot x^2 / ((2j + 2)(2j + 3)), \\
  \mathbf{S}(j + 1, x) &= \mathbf{S}(j, x) + \mathbf{t}(j + 1, x)
\end{align*}
with initial conditions $\mathbf{t}(0, x) = x$ and $\mathbf{S}(0, x) = x$ can be converted to computation by while loops. We do so in the definition of $\mathtt{sin}$, shown in \cref{fig:sin-def}.
Notice how the stopping condition for the while loop uses nondeterminism to ensure totality. It may stop if the error bound $\mathtt{abs}(t)$ is smaller than $2^{-n}$, or continue if it is larger than $2^{-n-1}$.
\begin{figure}
  \centering
  \begin{mdframed}
  \small
\begin{lstlisting}
let $\mathtt{sin}$(x:$\dR$):$\dR$ :=
  lim $n$.
    var j := $\numeral{0}$ in
    var S := $x$ in
    var t := $-x \times x \times x/\ccoerce{\numeral{6}}$ in
    while
      (case $2^{\iminus n \iminus \numeral{1}} \rlt \mathtt{abs}(t)$=>true | $~\mathtt{abs}(t) \rlt 2^{\iminus n}$=>false)
    do
      j := j + $\numeral{1}$ ;
      S := S + t ;
      t := $-t \times x \times x/\ccoerce{\numeral{2} \imult j \iplus \numeral{3}}$
    end ;
    $S$
\end{lstlisting}
  \end{mdframed}
  \caption{The definition of $\mathtt{sin}$}
  \label{fig:sin-def}
\end{figure}

The correctness of the implementation is outlined in \cref{fig:sin-correct}.
The loop invariant, abbreviated as $\phi$, is
$j \geq 0 \land S = \mathbf{S}(j, x) \land t = \mathbf{t}(j + 1, x)$.
Upon exit from the loop, the invariant ensures that $S$ is the $j$-th partial sum, and the exit condition $|t| < 2^{-n}$ that the error is sufficiently small. Together they imply that $S$ is within $2^{-n}$ of $\sin(x)$, as required by \rref{Ro-Lim}.
We elide several side conditions, such us those imposed by \rref{Ro-Lim} and \rref{Rw-Tot-Case}.
The complete details of the formal proof are available in the Coq formalization, see \cref{sec:formalization}.

\begin{figure}[htb]
  \centering
  \begin{mdframed}
  \small
\begin{lstlisting}
[\claimt{\top}]
lim $n$.
  [\claimt{\top}]
  var j := $\numeral{0}$ in
  var S := $x$ in
  var t := $-x \times x \times x/\ccoerce{\numeral{6}}$ in
  [\claimt{\phi}]      [where $\phi$ is $j \geq 0 \land S = \mathbf{S}(j,x) \land t = \mathbf{t}(j + 1, x)$]
  while
    [\claimt{\phi}]
    (case $2^{\iminus n \iminus \numeral{1}} \rlt \mathtt{abs}(t)$=>true | $~\mathtt{abs}(t) \rlt 2^{\iminus n}$=>false)
    [\claimt{b : \dB\such \phi \land (b \lthen 2^{-n-1} < \abs{t}) \land (\neg b \lthen \abs{t} < 2^{-n})}]
  do
    [\claimt{\phi \land 2^{-n-1}< \abs{t}}]
    [\claimt{j \geq 0 \land S = \mathbf{S}(j,x) \land t = \mathbf{t}(j + 1, x)}]
    j := j + $\numeral{1}$ ;
    [\claimt{j - 1 \geq 0 \land S = \mathbf{S}(j-1,x) \land t = \mathbf{t}(j, x)}]
    S := S + t ;
    [\claimt{j - 1 \geq 0 \land S = \mathbf{S}(j,x) \land t = \mathbf{t}(j, x)}]
    t := $-t \times x \times x/\ccoerce{\numeral{2} \imult j \iplus \numeral{3}}$
    [\claimt{j - 1 \geq 0 \land S = \mathbf{S}(j,x) \land t = \mathbf{t}(j+1, x)}]
    [\claimt{\phi}]
  end ;
  [\claimt{\phi \land \abs{t} < 2^{-n}}]
  [\claimt{S = \mathbf{S}(j, x) \land \abs{\mathbf{t}(j + 1, x)} < 2^{-n}}]
  $S$
  [\claimt{z  :  \dR\such \abs{z - \sin(x)} < 2^{-n}}]
[\claimt{y  :  \dR\such y = \sin(x)}]
\end{lstlisting}
  \end{mdframed}
  \caption{Specification of $\mathtt{sin}$.}
  \label{fig:sin-correct}
\end{figure}

\paragraph*{Computation of $\pi$}

We compute $\pi$ by using a root-finding algorithm to find the unique root of $\sin$ in the closed interval $[3,4]$.
For this, we could use the general root-finding algorithm in~\cite{brausse2016semantics}, which repeatedly narrows the search interval by splitting it into two overlapping intervals each $2/3$ the width of the original interval.
However, the fact that  we are computing a root that is known to be irrational, allows us to instead use a more efficient bisection-based search.
Once again, we proceed by computing a limit whose $n$-th term approximates the root within $2^{-n}$.
The code is shown in \cref{fig:pi-def}.
The bisection method is initialized with the lower bound $l = 3$ and the upper bound $u = 4$.
At each step of the iteration, we compute the midpoint $m = \frac{l + u}{2}$ and narrow the search interval to either $[m, u]$ or $[l, m]$, depending on the sign of $\sin(m)$. The irrationality of the  root guarantees that the comparison $0 < \sin(m)$ always terminates, because $m$ is rational and so $\sin(m) \neq 0$.
The above considerations are incorporated into the loop invariant
\begin{equation*}
  l \in \IQ \land u \in \IQ \land 3 \leq l < \pi < u \leq 4 \land u - l= 2^{-k}.
\end{equation*}
We leave detailed verification of correctness as exercise for the reader, who can also avoid doing it by consulting the formalized proof~\cite{clerical_coq}.
\begin{figure}
  \centering
  \begin{mdframed}
  \small
\begin{lstlisting}
lim $p$.
  var k := 0 in
  var l := $\ccoerce{\numeral{3}}$ in
  var u := $\ccoerce{\numeral{4}}$ in
  while k < n do
    var m := $(l \rplus u) \rmult 2^{-\numeral{1}}$ in
    if $\ccoerce{\numeral{0}} \rlt \mathtt{sin}(m)$ then
      l := m
    else
      u := m
    end ;
    k := $k \iplus \numeral{1}$
  end ;
  $(l \rplus u) \rmult 2^{\iminus \numeral{1}}$
\end{lstlisting}
  \end{mdframed}
  \caption{Computation of $\pi$.}
  \label{fig:pi-def}
\end{figure}

\section{Implementation}
\label{sec:implementation}

We turn attention to how Clerical, or a language based on it, might be implemented in practice.
A sensible implementation ought to work in such a way that an error-free well-typed \defemph{program} (a closed expression) $e$ of type~$\tau$ evaluates to one of its denotations, i.e., if $\sem{\emptyctx \rotypes e : \tau}() \neq \emptyset$ then $e$ evaluates to any $v \in \sem{\emptyctx \rotypes e : \tau}()$. More precisely, $e$ ought to evaluate to a \emph{representation} of~$v$, with the proviso that~$\bot$ corresponds to a non-terminating evaluation.

An implementation is certainly possible in principle. To see this, we can follow the approach of~\cite{brausse2016semantics} to show that Clerical programs are computable in the sense of Type Two Effectivity~\cite{w00}. We first endow each datatype~$\tau$ with a standard Baire space representation. In particular, reals are encoded by rapidly converging sequences of (encoded) rationals. We claim that whenever $\Gamma \rotypes e : \tau$ there is a type two Turing machine~$M$ which takes as input a representation of $\gamma \in \sem{\Gamma}$ and
\begin{itemize}
\item either $\sem{\Gamma \rotypes e : \tau} \, \gamma = \emptyset$, or
\item $M$ diverges and $\bot \in \sem{\Gamma \rotypes e : \tau} \, \gamma$, or
\item $M$ outputs a representation of an element of $\sem{\Gamma \rotypes e : \tau} \, \gamma$.
\end{itemize}
The construction of~$M$ proceeds recursively on the structure of~$e$. The most interesting is the $\ccase$ statement, which is implemented by combining the machines that compute the guards into a single machine that interleaves the guard computations and proceeds with the case whose guard first evaluates to~$\ctrue$.

More interesting and relevant is the question of an actual implementation of Clerical.
We implemented a proof-of-concept interpreter in OCaml, available at~\cite{clerical_ocaml}.
The connoisseurs will recognize the strong influence of the iRRAM package for exact-real arithmetic~\cite{muller2000irram}, which we gladly acknowledge.

We use the GNU MPFR library~\cite{mpfr} to compute with large integers and multiple-precision floating-point numbers.
During evaluation, real numbers are approximated by intervals whose endpoints are represented by multiple-precision floating-point numbers, rounded at the current \defemph{working precision}.
As the computation progresses, rounding, interval arithmetic, and limit approximations contribute to making the intervals ever wider.
If the intervals approximating $x$ and $y$ in a comparison $x \rlt y$ overlap, its Boolean value cannot be computed and evaluation is aborted due to \defemph{loss of precision}. The control returns to the top level, where the entire computation is restarted with higher working precision. The computation of a reciprocal $x^{-1}$ behaves analogously in case the interval approximating~$x$ overlaps with~$0$.
If the top level needs to output the result of a computation at a given precision, it keeps restarting it with ever higher working precision until the desired output precision is achieved.

The most interesting part of the interpreter is the implementation of guarded case
\begin{equation*}
  \ccase e_1 \To c_1 \mid \cdots \mid e_n \To c_n \cend.
\end{equation*}
Our semantics demands that one of the branches~$c_j$ must evaluate so long as one of the guards necessarily evaluates to $\ctrue$. Therefore, we cannot evaluate the guards $e_1, \ldots, e_n$ sequentially one after the other, lest the evaluation get stuck in a non-terminating guard.
We employed algebraic effects and handlers~\cite{plotkin09:_handl_algeb_effec,bauer15:_progr}, which are supported in OCaml~5, to implement cooperative threads that interleave the computations of the guards. The threads periodically yield control to the main scheduler, which enforces a round-robing evaluation strategy.
As soon as one of the threads~$e_i$ evaluates to~$\ctrue$, the other ones are aborted and the computation proceeds with the corresponding case~$c_i$.
We are looking forward to future experimentation with parallel execution of guards, which can be implemented using OCaml~5 multi-core features~\cite{sivaramakrishnan22:_retrof_concur}.


\section{Formalization}
\label{sec:formalization}

We formalized Clerical in the Coq proof assistant, and proved important properties presented in this paper. 
The formalization can also be used to formally type-check and verify Clerical expressions with respect to given specifications.
This has been done for the examples from \cref{sec:example}, their correctness proofs included.
Here we outline the organization of the formalization, which is freely available at~\cite{clerical_coq}:
\begin{description}[style=nextline,font=\normalfont]
\item[\coqpath{BaseAxioms}]
  Coq's type of propositions \coqcode{Prop} is intuitionistic, but we work in a classical
  metatheory. Thus we assume excluded middle for \coqcode{Prop}, dependent choice, as well as function extensionality and propositional extensionality.

\item[\coqpath{Powerdomain}]
  This part formalizes the domain-theoretic aspects of the powerdomain, as well as auxiliary constructions
  required in the formalization of the denotational semantics.

\item[\coqpath{Syntax}, \coqpath{Typing}, \coqpath{TypingProperties}]
  These parts formalize the syntax and typing rules of Clerical from \cref{sec:syntax}.
  We use de Bruijn indices to implement variables. We also prove several properties of the
  type system, such as uniqueness of typing (an expression has at most one type).

\item[\coqpath{Semantics}, \coqpath{SemanticProperties}]
  The formalization of the denotational semantics relies on the formalization of the powerdomain and uses the real numbers from Coq's standard library as the denotation of the type of reals.
  The denotation of a term proceeds by induction of its typing derivation. We prove that the denotation does not
  depend on the derivation, but only on the term.

\item[\coqpath{Specification}, \coqpath{ReasoningRules}, \coqpath{ReasoningSoundness}, \coqpath{ReasoningAdmissible}]
  These modules formalize Hoare triples, their reasoning rules, and show them to be sound with respect
  to the denotational semantics. We also show that the rules given in \cref{sec:admissible-rules} are admissible.

\item[\coqpath{examples}, \coqpath{Mathematics}, \coqpath{Arith}]
  We formalized the examples from \cref{sec:example}, where we assumed several familiar facts about~$\pi$ in \coqpath{Mathematics}, for example that if $3 < x < 4$ and $\sin x = 0$ then $x = \pi$.

\item[\coqpath{Arith}]
  In the formalized syntax of Clerical, every expression must be well-typed, and consequently equipped with
  a formal typing derivation. Constructing these is straightforward, but creating them by hand is time-consuming.
  We therefore provided automation for constructing well-typed arithmetical expressions involving variables,
  constants, arithmetical operations, comparisons, and coercions.
  We additionally automated proofs showing that such expressions satisfy the expected partial and total correctness specifications.
  The automation is used extensively in the formalized examples.
\end{description}

\section{Future work}
\label{sec:future-work}

In conclusion we discuss several directions for further work.

One is to explore how Clerical could be extended to include higher-order functions and general recursion.
Incorporating higher-order function without recursion should be straightforward from a language-design viewpoint. However, generalising the denotational semantics to cover higher-order functions may not be so straightforward, since the powerdomain $\PP{S}$ from 
Section~\ref{sec:denotation} will have to be generalised to allow $S$ to range over denotations of arbitrary types. 
With regards to general recursion, the non-monotonicity phenomena associated with the guarded case construct are likely to make it very challenging to define denotational semantics; see~\cite{LEVY2007221} for a discussion of related issues.

Second, in this paper we have not presented any formal operational semantics for Clerical.
Having one would provide an alternative and direct
account of the computability of the language, as well as a
framework within which implementation-relevant information, such as
the scope for parallelism in the execution strategy, could be studied in a mathematical setting. Also, a formally specified operational semantics
could guide implementations of Clerical and Clerical-like languages, and help estaliblish their correctness.

Third, we could further experiment with our implementation, which is good enough to evaluate the~$\pi$ program but cannot compete with the mature libraries for exact-real numbers.
To speed it up, we should at least implement parallel execution of threads, which is supported by the latest OCaml version.
A more substantive improvements would explore better evaluation strategies for nondeterminism, and compilation to a more efficient low-level language.

Fourth, there is significant room for improvement in the Coq formalization. By implementing better automation and tactics for proving correctness assertions, we would obtain a workable environment for formal verification of exact real computation, supported by the formidable machinery of Coq.

\bibliographystyle{alpha}
\bibliography{references.bib}

\end{document}